\documentclass{article}

\usepackage[english]{babel}

\usepackage[letterpaper,top=2cm,bottom=2cm,left=2.5cm,right=2.5cm,marginparwidth=1.75cm]{geometry}

\usepackage{amsmath}
\usepackage{amssymb}
\usepackage{amsthm}   

\numberwithin{equation}{section} 

\newtheorem{theorem}{Theorem}[section]
\newtheorem{proposition}[theorem]{Proposition}
\newtheorem{lemma}[theorem]{Lemma}
\newtheorem{corollary}[theorem]{Corollary}

\newtheorem{remark}[theorem]{Remark}

\usepackage{graphicx}
\usepackage[colorlinks=true, allcolors=blue]{hyperref}
\usepackage{xcolor}

\definecolor{darkgreen}{RGB}{0,100,0}  

\usepackage{subfig} 
\graphicspath{{images/}} 
\usepackage{todonotes}

\def\R{\mathbb{R}}
\def\C{\mathbb{C}}
\def\eps{\varepsilon}
\def\rmd{\mathrm{d}}
\def\rmi{\mathrm{i}}
\def\rme{\mathrm{e}}
\def\calO{\mathcal{O}}

\definecolor{colorJRblue}{rgb}{0.,0.,1.}

\definecolor{colorLMred}{rgb}{1.,0.,0.}

\usepackage{authblk}

\title{Rolls and Snaking in a Swift-Hohenberg Equation\\ with Non-smooth Nonlinearity}
\author{L\"utfiye Masur} 
\author{Jens D. M. Rademacher} 

\affil{Department of Mathematics, University of Hamburg, 20146 Hamburg, Germany\\ \texttt{luetfiye.masur@uni-hamburg.de}\\ \texttt{jens.rademacher@uni-hamburg.de}}

\begin{document}
\maketitle

\begin{abstract}
We study rolls and homoclinic snaking  in a variation of the one-dimensional Swift-Hohenberg equation, whose standard forms are prototypical order-parameter models for pattern formation in the sciences. 
Motivated by classes of differential equation models that involve continuous non-smooth  
low order nonlinear terms, we replace the standard quadratic-cubic nonlinearity by 
$\nu |u|^\alpha-u^3$, $\alpha \in [1,2]$ with $\nu > 0$.  
In the vicinity of zero, for $\alpha<2$ this nonlinearity falls outside the scope of classical Taylor expansion and bifurcation analysis. 
Our partially analytical and partially numerical results highlight that the non-smooth term modifies the criticality of pattern-forming bifurcations and alters the associated branches of solutions. In particular, $\alpha\in(1,2)$ implies subcriticality of roll bifurcations for any $\nu>0$. At $\alpha=1$ differentiability is lost, which has a strong impact on the bifurcations of sign-changing rolls including the disappearance of homoclinic snaking. Homoclinic snaking thus emerges non-smoothly as $\alpha$ increases from $\alpha=1$, and persists when retaining an additional, e.g., quadratic term.

\end{abstract}


\tableofcontents

\section{Introduction}\label{sec:intro}
Motivated by the appearance of low order non-smooth nonlinearities in various modelling contexts, in this paper we study Swift-Hohenberg-type equations for an order parameter $u = u(x, t) \in \mathbb{R}$, $x \in \mathbb{R}$ with such terms.  
Our focus is on the impact on roll bifurcation criticality thus leading to the form
\begin{equation} \label{nonsmooth SHE}
	\partial_t u= -(1+\partial_x^2)^2 u + \mu u + \nu |u|^{\alpha} -u^3,  \quad \alpha \in [1,2],
\end{equation}
with fixed parameter $\nu \in\mathbb{R}$ and bifurcation parameter $\mu \in \mathbb{R}$. The classical Swift-Hohenberg equation in \cite{SwiftHohenberg1977} has $\nu=0$ and was originally introduced in the study of thermal fluctuations in a fluid near the Rayleigh-B\'enard instability. Here, the linear operator at $\mu=0$ selects a preferred (critical) wavelength and the cubic term provides nonlinear saturation, thereby bounding the amplitude of emerging solutions. A prominent generalisation features quadratic-cubic nonlinearity in the form
\begin{equation}\label{smooth SH}
	\partial_t u = -(1 + \partial_x^2)^2 u + \mu u + \nu u^2 -u^3, 
\end{equation}
where the quadratic nonlinearity breaks the reflection symmetry $u \rightarrow -u$ and via the interaction of harmonics  
can produce hysteresis and coexistence of stable patterned and homogeneous states.

Indeed, the generalised Swift-Hohenberg equation \eqref{smooth SH} serves as a prototypical model for pattern formation and has been extensively studied, especially concerning spatially periodic and localized steady states. In this broader context, it qualitatively models nonlinear pattern formation in a wide range of systems, including fluid convection, chemical reactions, biological systems, and nonlinear optics, e.g., \cite{CrossHohenberg1993,CrossGreenside2009}. It is well known that $\nu^2>\tfrac{27}{38}$ yields a subcritical regime, admitting bistability of the zero state and $2\pi$-periodic rolls, which leads to spatially localized states that consist of finite patches of rolls embedded in a homogeneous background. This can be understood via front pinning near Maxwell points, where homogeneous and patterned states have equal energy, which gives rise to localized states organized by homoclinic snaking \cite{BurkeKnobloch2006,BurkeKnobloch2007}. A geometric explanation of this phenomenon from the spatial ODE is provided by the theory of heteroclinic tangles and invariant manifold intersections developed in \cite{WoodsChampneys1999}, which describes how snaking arises from the unfolding of a degenerate reversible Hamiltonian--Hopf bifurcation. The resulting snakes-and-ladders bifurcation structure describes the existence and stability of localized patterns \cite{Beck2009}.

\medskip
Our motivation for the variant \eqref{nonsmooth SHE} stems from models that involve non-smooth nonlinearities, in particular with low order. In quantum droplet models, nonlinear terms of the form $|\psi|\psi$ arise naturally \cite{KatsimigaMistakidis2023,SaqlainMithun2023} and in hydrodynamic drag laws, where quadratic bottom drag in fact yields a Swift--Hohenberg-type equation with a $|\psi|\psi$ nonlinearity  \cite{PruggerRademacherYang2023}. More generally, drag forces in systems are modelled by terms of the form $v|u|$ for velocity components $u,v$, e.g.\ \cite{ZazoRdemacher2020,ZazoRademacher2023} and the references therein. While these nonlinearities are quadratic, the latter type includes piecewise linear behaviour. Also reaction--diffusion models feature piecewise-linear terms, e.g. \cite{ZemskovTsyganov2019,Bakanas2003}, where, however, the considered solutions do not pass through non-smooth parts. Generalised Swift--Hohenberg equations can be viewed as simplified models for reaction-diffusion systems with Turing instability, and thus provide a natural simplified context for investigating the impact of such nonlinearities on pattern formation. From an energetic viewpoint, a Swift-Hohenberg-type equation with nonlinearity $F':\R\to\R$ is an $L^2$-gradient flows with energy density $\big( (1 + \partial_x^2) u \big)^2 - \mu u^2 -F(u)$. Hence, another motivation for a term $|u|^\alpha$ in $F'$ is that it results from generalised potentials $F$ that contain the term $-u|u|^{\alpha}/(1+\alpha)$.

\medskip
Classical bifurcation theory relies on smooth nonlinearities and the validity of a Taylor expansion. For nonlinearities of the form $|u|^\alpha$, smoothness at $u=0$ is lost for $\alpha<2$, while for $\alpha=1$ it is piecewise linear and only Lipschitz continuous. In this paper study bifurcation phenomena of rolls and related localised steady states, in particular for the transition to $\alpha=1$. The differentiability for $\alpha\in(1,2]$ indeed admits a rigorous analysis of rolls bifurcating from $u=0$ using suitable estimates adapted to the non-smooth setting (Theorem~\ref{thm:rolls1}):  the bifurcation is subcritical for \emph{any} $\nu>0$ with convex, increasingly flat branch as $\alpha$ decreases below $\alpha=3/2$. See Figure~\ref{fig:intro1}(a). Numerically, this yields a widening bistable region, but we find that the pinning region of snaking in fact shrinks for decreasing $\alpha$, and the amplitude of energetically compatible rolls decreases. See Figure~\ref{fig:intro1}(b). (We remark that an absence of a threshold for subcriticality also occurs in the smooth cubic-quintic Swift-Hohenberg equation \cite{BurkeKnobloch2007}.) 

		\begin{figure}
			\centering
			\begin{tabular}{cc}
			\includegraphics[width=0.45\textwidth]{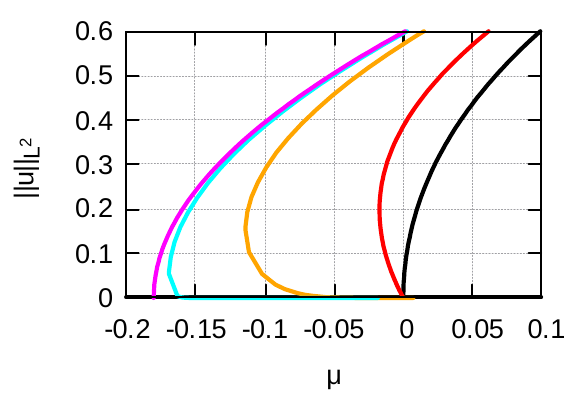} &
			\includegraphics[width=0.41\linewidth]{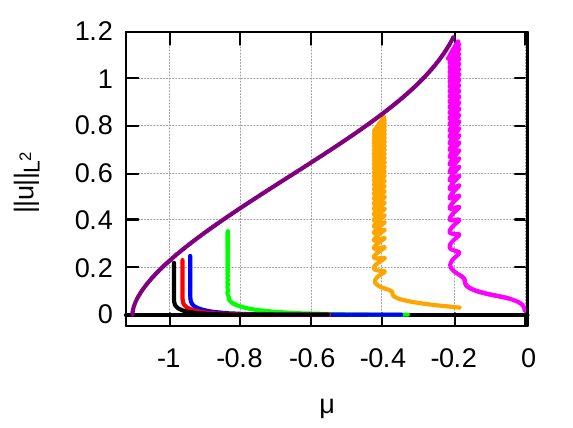}  \\
			(a) & (b)
			\end{tabular}
\caption{(a) Branches of $2\pi$-periodic rolls of \eqref{nonsmooth SHE} for $\nu=0.5$ at $\alpha = 2$ (black), $\alpha = 1.5$ (red), $\alpha = 1.1$ (orange), $\alpha = 1.01$ (cyan), $\alpha =1$ (magenta). (b) Snaking branches  of \eqref{nonsmooth SHE} for $\nu = 1.6$ and $\alpha = 2$ (magenta), $1.5$ (orange), $1.1$ (green), $1.05$ (blue), $1.04126$ (red), $1.03242$ (black) and the scaled $L^2$-norm of the rolls at the Maxwell points $\mu=\mu_M(\alpha,1.6)$ (purple) for $\alpha \in [1,2]$.}
			\label{fig:intro1}
		\end{figure}

\medskip
The limiting value $\alpha=1$ does not admit a linearisation at $u=0$ and thus even the location of bifurcation points is unclear. For merely Lipschitz ODE local invariant manifold theory still (abstractly) applies to some extent \cite{AulbachWanner1996}, and more specifically for piecewise linear systems in combination with Hopf-bifurcation analysis \cite{Hosham2018,PonceRosVela2022,KuepperHoshamWeiss2013,KuepperHosham2011}; see also \cite{ZazoRdemacher2020,ZazoRademacher2023}. However, these results for certain (non-reversible) Hopf bifurcations do not seem to constructively help our specific case, although our analysis shares similarities. For Hopf bifurcations in piecewise smooth planar systems we also refer to \cite{Simpson2022}, and more broadly for non-smooth dynamics and bifurcations to \cite{diBernardo2007}, as well as the more recent review \cite{BelykhKuskePorfiriSimson2023}. For perturbations from smooth systems, reduced dynamics can be studied via so-called spectral submanifolds \cite{BettiniCenedeseHaller2024}, which relates in our case to $\nu\approx 0$. 

Rather than building on abstract results, for $\alpha=1$ on the one hand we use energy estimates, and on the other hand analyse the bifurcation of sign-changing spatially periodic solutions directly by using linearity for $u>0$ and $u<0$, and a suitable boundary value problem. From energy estimates it is straightforward to rule out bifurcations for $\mu\leq -\nu$ (Lemma~\ref{lem:globaldecay}) and infer instability of the zero state for $\mu> 0$ (Remark~\ref{rem:energynondecay}). Indeed, we find that the primary bifurcation points $\mu^*(\nu)$ of sign-changing rolls occur for $-\nu< \mu < 0$. Here we use that for $\mu>-\nu$ the boundary value problem can be reduced to a system of two transcendental algebraic equations that differ for each configuration of characteristic roots for $u>0$ and $u<0$ (cf.\ Theorem~\ref{theorem: reduced system nu>1/2}). 

We numerically study the algebraic systems and find a unique branch of rolls in this range, proving this rigorously for $\nu\approx 0$ (Theorem~\ref{thm:smallnu}). Here we omit the cubic term for simplicity and expect the branch perturbs when including higher order terms analogous to \cite{BettiniCenedeseHaller2024}. Note that for $\nu |\cdot|$ any convex smoothening has vanishing derivative at $u=0$ and thus incorrectly yields the primary roll bifurcation at $\mu=0$. 

Interesting  effects occur at $\alpha=1$ concerning the periods of bifurcating rolls: Although there is no linearization that readily justifies a critical wavenumber, we numerically find that when increasing $\mu$ the primary rolls to bifurcate are $2\pi$-periodic, as suggested by the maxima of the Fourier-symbol of $-(1+\partial_x^2)^2$, albeit with unclear $\mu$-value (\S\ref{subsec: different periods}). However, in contrast to the case $\alpha>1$, where periodic solutions bifurcate continuously in terms of $\mu$ after the $2\pi$-periodic solution, for $\alpha=1$ the possible periods of the rolls are divided into two regions of $\mu$-values: in one region, the period is bounded from below by $\sqrt{2}\pi$, while in the other, it is bounded from above by $\sqrt{2}\pi$. At the critical period $\sqrt{2}\pi$, the rolls transition from sign-changing to non-sign-changing solutions (Theorem~\ref{thm:lowl}).

Overall, our results provide a first systematic examination on the effect of the regularity parameter $\alpha\in[1,2]$ in \eqref{nonsmooth SHE}, with particular focus on how it modifies the fundamental bifurcation mechanisms of rolls and homoclinic snaking.

\medskip
This paper is organised as follows. In \S\ref{sec:homrolls}, we study the spatially homogeneous steady states associated with the nonlinearity $\nu |u|^\alpha-u^3$, and, for $\alpha\in(1,2]$, the corresponding spatial eigenvalues and their relation to periodic roll solutions (\S\ref{sec:homostatesevals}). For $\alpha=1$ this includes rolls that do not change sign. In \S\ref{sec:rolls}, we  investigate the bifurcation of sign-changing rolls.  
In \S\ref{subsec: Bifurcation of periodic rolls} we consider $\alpha=1$ and perform the aforementioned reduction to algebraic equations, complemented by numerical computations. In \S\ref{subsec: different periods}, we analogously investigate roll solutions with different wavenumbers for $\alpha=1$.

Finally, in \S\ref{sec:snaking} we study homoclinic snaking. \S\ref{subsec: homoclinic snaking} examines this through the associated energy functional and Maxwell point analysis. In contrast to the smooth case, snaking appears to be absent when $\alpha=1$. To further investigate this phenomenon, \S\ref{subsec: homoclinic snaking at 1} concerns a modified model with an additional term designed to restore subcriticality and facilitate the existence of localized states at $\alpha=1$.

\section{Spatial ODE, homogeneous states and accompanying rolls}\label{sec:homrolls}
	
	Since \eqref{nonsmooth SHE} is invariant under the reflection symmetry $(\nu,u) \rightarrow (-\nu,-u)$, it suffices to restrict attention to the case $\nu >0$. The spatial profiles of steady states solves the so-called spatial ODE given by \eqref{nonsmooth SHE} with zero left hand side,
\begin{equation} \label{e:spatialNSSHE}
 -(1+\partial_x^2)^2 u + \mu u + \nu |u|^{\alpha} -u^3 =0.
\end{equation}
It is well known that 
 the stationary Swift-Hohenberg equation can be formulated as a Hamiltonian system in space, see \eqref{eq:energy_func}. In particular,  \eqref{e:spatialNSSHE} is a reversible ODE \cite{MR402815} through the reflection induced by $x\to-x$.

\subsection{Spatially Homogeneous Steady States}\label{sec:homostates}
			
		We first consider spatially homogeneous steady states of \eqref{nonsmooth SHE}, i.e., equilibria of \eqref{e:spatialNSSHE}.   
		In the classical case $\alpha = 2$ these are
		\begin{equation}\label{e:homeq}
			 u_0 = 0 \text{ and }  u_{\pm 0} = \frac{\nu}{2} \pm \frac{1}{2} \sqrt{\nu^2 -4(1-\mu)},
		\end{equation}
		which undergo,  at $\mu = 1$, a transcritical bifurcation and at $\mu = 1-\frac{\nu^2}{4}$ the nonzero branch undergoes a fold bifurcation. See Figure~\ref{figure: homogeneous states}(a).		
		\begin{figure}
			\centering
			\begin{tabular}{cc}
				\includegraphics[width=0.5\linewidth]{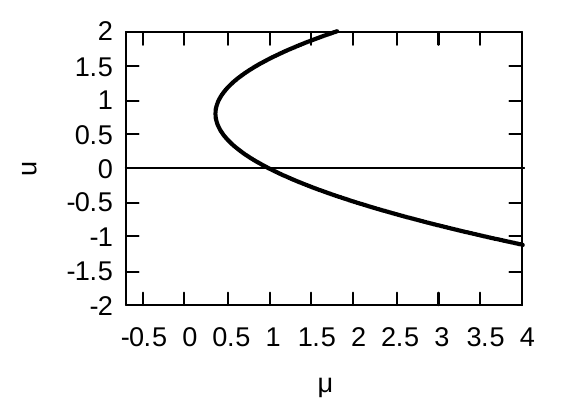} & \includegraphics[width=0.49\linewidth]{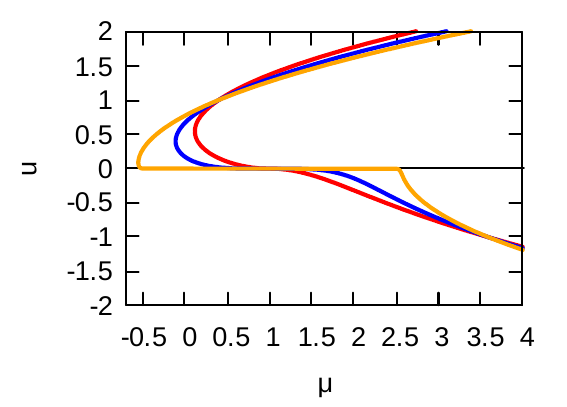} \\
				(a) & (b) 
			\end{tabular}
			\caption{ For $\nu = 1.6$: (a) Branches of spatially homogeneous equilibria in \eqref{e:homeq} for $\alpha=2$ . (b) Branches of spatially homogeneous equilibria of \eqref{eq: homeq_for_alpha} for $\alpha = 1.5 \text{ (red), } 1.25 \text{ (blue), } 1.01 \text{ (orange) }$.}
			\label{figure: homogeneous states}
		\end{figure}
	
\paragraph{The case $\alpha\in(1,2)$.} The spatially homogeneous steady states of \eqref{nonsmooth SHE} for $\alpha \in (1,2)$ 
satisfy
		\begin{equation}\label{eq: homeq_for_alpha}
			0 = -1 + \mu + \nu \text{sgn}(u) |u|^{\alpha-1}-u^2, 
		\end{equation} 
and it is convenient to solve this as 
		\begin{equation}
			\mu = \mu (u) = \begin{cases}
				1 - \nu u^{\alpha-1} + u^2, & u>0 \\
				1 + \nu (-u)^{\alpha-1} + u^2, & u<0 .
			\end{cases} 
		\end{equation} 
		In particular, $\mu (0) =1$, so the bifurcation point for the trivial state lies at $\mu = 1$ for all $\alpha \in (1,2)$. 
		
\noindent For $u>0$, differentiation gives
		\begin{equation}
			\mu '(u) = - \nu (\alpha-1) u^{\alpha-2} + 2u,
		\end{equation}
		and there is a unique critical point $\mu '(u) = 0$ at 
		\begin{equation}\label{e:homeqcrit}
			u_* = \biggl( \frac{\nu(\alpha-1)}{2} \biggr) ^{\frac{1}{3-\alpha}} .
		\end{equation}
		The corresponding parameter value gives a fold point (also for $\alpha=2$) at
		\begin{equation}
			\mu (u_*) = 1- \nu \biggl( \frac{\nu(\alpha-1)}{2} \biggr)^{\frac{\alpha-1}{3-\alpha}} + \biggl( \frac{\nu(\alpha-1)}{2} \biggr)^{\frac{2}{3-\alpha}} .
		\end{equation}   
		See Figure \ref{figure: homogeneous states}(b). For $\alpha=2$ this is the fold in Figure~\ref{figure: homogeneous states}(a) and the limiting value for $\alpha\to 1$ is given by
		\begin{equation}
			\lim_{\alpha\rightarrow 1} \mu(u_*) = 1- \nu \text{ with } \lim_{\alpha\rightarrow 1} u_* = 0 ,
		\end{equation}
		and corresponds to the left bifurcation point in Figure~\ref{figure: homogeneous for alpha 1}. \\
		\noindent For $u<0$, there are no fold points, but  
		\begin{equation}
			\mu ''(u)  =\nu (\alpha -1)(\alpha-2)(-u)^{\alpha-3} +2, 
		\end{equation}
		which has a root at 
		\begin{equation}
			u_*^- = - \biggl( - \frac{\nu (\alpha-1)(\alpha-2)}{2} \biggr) ^{\frac{1}{3-\alpha}} 
		\end{equation}
		with corresponding parameter value 
		\begin{equation}
			\mu (u_*^-) = 1+ \nu \biggl( - \frac{\nu(\alpha-1)(\alpha-2)}{2}\biggr)^{\frac{\alpha-1}{3-\alpha}} + \biggl( - \frac{\nu (\alpha-1)(\alpha-2)}{2} \biggr)^{\frac{2}{3-\alpha}} .
		\end{equation} 
		For $\alpha=2$ this lies at at the transcritical bifurcation point in Figure~\ref{figure: homogeneous states}(a), and the limit for $\alpha\to1$, 
		\begin{equation}
			\lim_{\alpha\rightarrow 1} \mu (u_*^-) = 1 + \nu \text{ with } \lim_{\alpha\rightarrow 1} u_*^- = 0 ,
		\end{equation}
		corresponds to the right bifurcation point in Figure~\ref{figure: homogeneous for alpha 1}. 

\paragraph{The case $\alpha=1$.}
		
The spatially homogeneous steady states for $\alpha = 1$ solve 
		\begin{equation}\label{e:homeqal1a}
			0 = - u + \big(\mu + \text{sgn(u)}\nu \big) u -u^3. 
		\end{equation}
		The term $\text{sgn}(u)\nu u$ breaks the symmetry so that one branch with $u>0$ emerges at $\mu = 1-\nu$ and another with $u<0$ at $\mu = 1+\nu$. These are explicitly given by 
		\begin{equation}\label{e:homeqal1}
			u_0 = 0, \quad u_{+0} = \sqrt{\mu + \nu -1} \text{ for } u>0, \quad u_{-0} = -\sqrt{\mu - \nu -1} \text{ for } u<0,
		\end{equation}
		so there is one supercritical `half-pitchfork' bifurcation at $\mu = 1-\nu$ and one at $\mu = 1+\nu$. See  Figure~\ref{figure: homogeneous for alpha 1}.

\subsection{Spatial eigenvalues and related roll solutions}\label{sec:homostatesevals}

\begin{figure}
	\centering
	\includegraphics[width=0.5\textwidth]{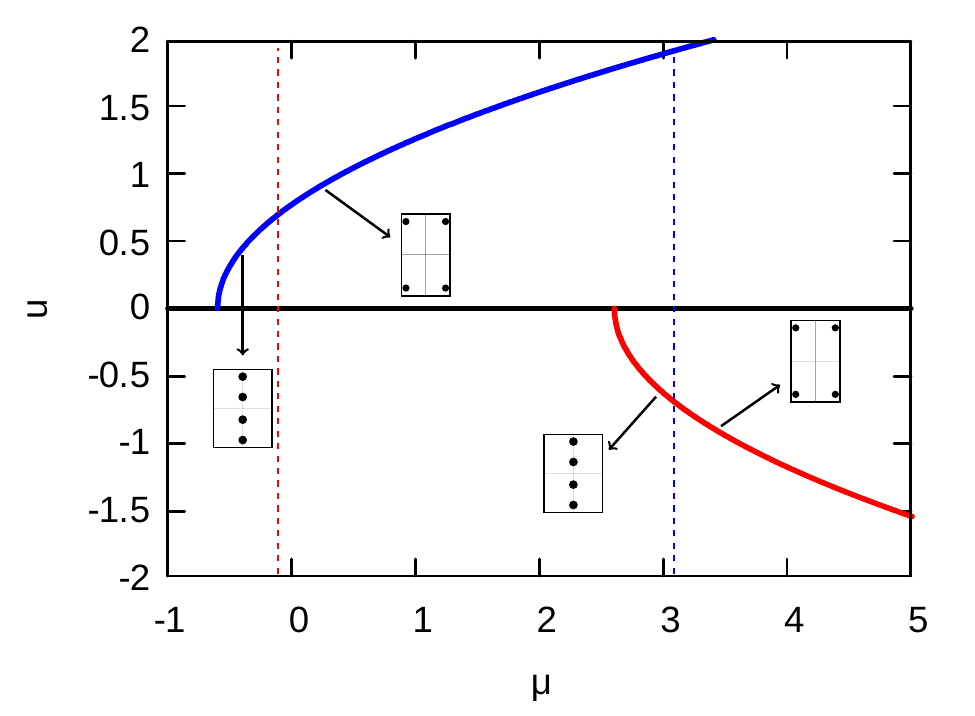}
	\caption{Spatially homogeneous solutions of \eqref{nonsmooth SHE} at $\alpha=1$ together with the spatial eigenvalues of the branches of $u_{+0}$ and $u_{-0}$ for $\nu = 1.6$, which are given in \eqref{e:homeqal1}. The red and blue dashed lines indicate the parameter values $3/2-\nu$ and $3/2+\nu$, where $u_{\pm0}$ undergo a reversible Hopf bifurcation, respectively. 
		Roll solutions accompany $u_\pm$  in \eqref{eq:u+-sol} , where the eigenvalues are purely imaginary (at least when non-resonant).}
	\label{figure: homogeneous for alpha 1}
\end{figure}

For $\alpha\in(1,2]$ the linearisation of \eqref{e:spatialNSSHE} in a spatially homogeneous equilibrium $u$ reads 
\begin{equation*}
	\mu v - (1+\partial_x^2)^2 v + \alpha\nu\, \mathrm{sgn}(u) |u|^{\alpha-1}v - 3u^2 v =0,
\end{equation*} 
whose characteristic equation in terms of $k\in\C$ can be written as
\begin{equation}\label{e:spatialNSSHEchar}
	\mu - (1-k^2)^2 + \alpha\nu\, \mathrm{sgn}(u_0) |u|^{\alpha-1} - 3u^2=0.
\end{equation}
The solutions are given by the two sign choices and the negative values in
\begin{equation}\label{e:kpm}
k_\pm = \sqrt{1 \pm \sqrt{\mu+ \alpha\nu\, \mathrm{sgn}(u) |u|^{\alpha-1} - 3u^2}}
\end{equation}
which are the complex frequencies of the spatial eigenvalues $\pm\rmi k_\pm$ of the linearized spatial ODE. It is well known that if $k_+$ is real and $k_-$ purely imaginary, then the reversible Lyapunov centre theorem \cite{MR402815} applied to \eqref{e:spatialNSSHE} implies the existence of a one-parameter family of periodic orbits to \eqref{e:spatialNSSHE} in a smooth two-dimensional manifold including $u$. Its wavenumber converges to $k_+$ as $u$ is approached within the family. If both $k_\pm$ are real and non-resonant, then in addition such a family exists whose wavenumber converges to $k_-$. See also \cite[\S 4.3.4]{HaragusIooss}. For $k_\pm\notin\R$ there are no such accompanying periodic orbits. The transition from non-zero real to complex spatial frequencies is a reversible Hopf-bifurcation for the spatial ODE, cf.\ e.g.\ \cite[\S 4.3.3]{HaragusIooss} with more complex bifurcations. It is a Turing-instability/-bifurcation for the PDE and any spatially periodic orbits correspond to roll solutions in \eqref{nonsmooth SHE}. The transition of one frequency through zero, while another remains non-zero is a reversible Hopf-zero bifurcation, here in 4D cf.\ e.g.\ \cite[\S 4.3.1]{HaragusIooss}, also with more complex bifurcations. For comparison we explicitly note the well-known situation away from these complications for $\alpha=2$, cf.\ \eqref{e:homeq}, which is qualitatively analogous for all $\alpha\in(1,2)$:
\begin{itemize}
\item At $u=0$ we have $k_\pm\notin \R$ for $\mu<0$, $k_\pm\in\R$ for $\mu\in(0,1)$, and $k_+\in\R$, $k_-\in\rmi\R$ for $\mu>1$. Hence, $u=0$ is accompanied by rolls for $\mu>0$: two families for $\mu\in(0,1]$ (at least whenever $k_\pm$ are non-resonant) and one for $\mu> 1$. 

\item At $u=u_{+0}$ for $u_{+0}\in (u_*,  \frac{\nu + \sqrt{\nu^2 + 8}}{4})$, cf.\ \eqref{e:homeqcrit}, we have $k_\pm\in\R$, and for $u>\frac{\nu + \sqrt{\nu^2 + 8}}{4}$ we have $k_\pm\notin\R$. Hence, $u=u_{+0}$ is accompanied by rolls for $u_{+0}\in (u_*,\frac{\nu + \sqrt{\nu^2 + 8}}{4})$ only, and then two families (at least whenever $k_\pm$ are non-resonant). 

\item At $u=u_{-0}$ for $u_{-0}\in(0,u_*)$ we have $k_+\in\R$, $k_-\in\rmi\R$, for $u_{-0}\in\left(\frac{\nu - \sqrt{\nu^2 + 8}}{4},0\right)$ we have $k_\pm\in\R$, and for $u_{-0}<\frac{\nu - \sqrt{\nu^2 + 8}}{4}$ we have $k_\pm\notin \R$. Hence, $u=u_{-0}$ is accompanied by rolls for 
$u_{-0}\in(\frac{\nu - \sqrt{\nu^2 + 8}}{4},u_*)$ only: one family for $u_{-0}\in(0,u_*)$ and two otherwise (at least whenever $k_\pm$ are non-resonant).

(The ranges in terms of $\mu$ for latter two can be computed as $\mu\in(\mu(u_*),1)$, and $(1,\frac{12-\nu^2+\nu\sqrt{\nu^2+8}}{8})\big)$.)
\end{itemize}
\medskip
For $\alpha=1$, due to the lack of differentiability in $u=0$, one cannot linearize and there are no (unique) spatial eigenvalues for $u=0$ so that the Lyapunov centre theorem cannot be applied. We discuss this issue in more detail in \S\ref{subsec: Bifurcation of periodic rolls}. 
However, sufficiently close to the branches $u_{\pm0}$ from \eqref{e:homeqal1} such that $u\neq 0$ has fixed sign, this problem does not arise. 
In these cases \eqref{nonsmooth SHE} with $\tilde{\mu} = \mu+\nu$ for $u>0$ and $\tilde{\mu}=\mu-\nu$  for $u<0$ reads
\begin{equation}\label{eq:SHcubic}
	\partial_t u = -(1+\partial_x^2)^2 u + \tilde{\mu} u - u^3.
\end{equation} 
Hence, up to a shift in the bifurcation parameter, the situation is that of the classical cubic SHE. Without loss of generality we consider $u_{+0}= \sqrt{\tilde{\mu} -1}$ for $\tilde{\mu}>1$ and obtain linearized equation 
		\begin{equation*}
			\partial_t v = \tilde{\mu} v - (1+\partial_x^2)^2 v - 3(u_{+0})^2 v  = - (1+\partial_x^2)^2 v - (2\tilde{\mu}- 3)v.
		\end{equation*} 
 The corresponding characteristic equation with $\tilde\mu=\mu+\nu$ is
		
		\begin{equation}\label{eq:u_+0 char eq}
			\mu = -\frac{1}{2} (1-k^2)^2 + \frac{3}{2} - \nu
			\end{equation}
which yields the spatial eigenvalues $\rmi k$ with the complex wavenumbers $k$ 
\begin{equation}
k = \pm \tilde k_\pm, \quad \tilde k_\pm:=\sqrt{1 \pm \sqrt{3-2\mu-2\nu} }.
\end{equation} 
These are real valued for $1-\nu < \mu \leq 3/2-\nu$ and $\tilde k_\pm\notin\R$ for $\mu> 3/2-\nu$; recall $\mu>1-\nu$ for $u_{+0}$ to exist. Figure~\ref{figure: homogeneous for alpha 1} displays the corresponding spatial eigenvalues along the $u_{\pm0}$ branches. 
Hence, as above, for each parameter value $1-\nu < \mu < 3/2-\nu$ at which $k_\pm$ are rationally independent, $u_{+0}$ is accompanied by roll solutions in $\{u>0\}$. We note that at $\mu = 3/2-\nu$ a reversible Hopf bifurcation occurs with critical wavenumber $k_c=1$.

When the $u_{+0}$-branch connects with $u=0$ at $\tilde\mu=1$, an aforementioned reversible Hopf-zero bifurcation occurs in \eqref{eq:SHcubic}. This admits a branch of rolls that switch from sign-changing to non-sign-changing \cite[\S 4.3.1]{HaragusIooss} and thus yields a branch of admissible rolls for $\mu>1-\nu$; in \S\ref{subsec: different periods} we discuss its connection with an admissible branch for $\mu<1-\nu$.

We analogously obtain that roll solutions in $\{u<0\}$ accompany the $u_{-0}$ branch for $\mu \in (1+\nu,3/2+\nu)$ (at least where $\tilde k_\pm$ are non-resonant).

\section{Sign changing roll solutions}\label{sec:rolls}
	
In this section we examine roll solutions near $u=0$ in more detail. In \S\ref{sec:homostatesevals} we already noted their existence for $\alpha\in (1,2]$ as a consequence of spatial eigenvalues and here we discuss their bifurcations with focus on wavenumber $k=1$. In Figure~\ref{figure: periodic for different alphas}(a) we plot the spatial eigenvalues at $u=0$. When increasing $\mu$, spatial eigenvalues become purely imaginary at $\mu=0$ with wavenumber $1$, which corresponds to a Turing instability and thus we expect the emergence of $2\pi$-periodic rolls. In Figure~\ref{figure: periodic for different alphas}(a) we additionally plot, for $\alpha=2$, branches of homogeneous equilibria and $2\pi$-periodic rolls computed by numerical continuation. This highlights the well-known subcriticality of the bifurcation to rolls in the quadratic-cubic SHE if $\nu^2 >\frac{27}{38}$; here we take $\nu=1.6$. In \S\ref{s:rollsalphabigger1} we analyse the criticality of this bifurcation for $\alpha\in(1,2]$. 

In Figure~\ref{figure: periodic for different alphas}(b) we plot numerical continuations of these branches outside the scope of weakly nonlinear analysis for decreasing values of $\alpha$ and also for $\alpha=1$. We analyse the corresponding bifurcation to rolls in the most degenerate case $\alpha=1$ in \S\ref{subsec: Bifurcation of periodic rolls}.
		
		\begin{figure}
			\centering
			\begin{tabular}{cc}
				\includegraphics[width=0.47\linewidth]{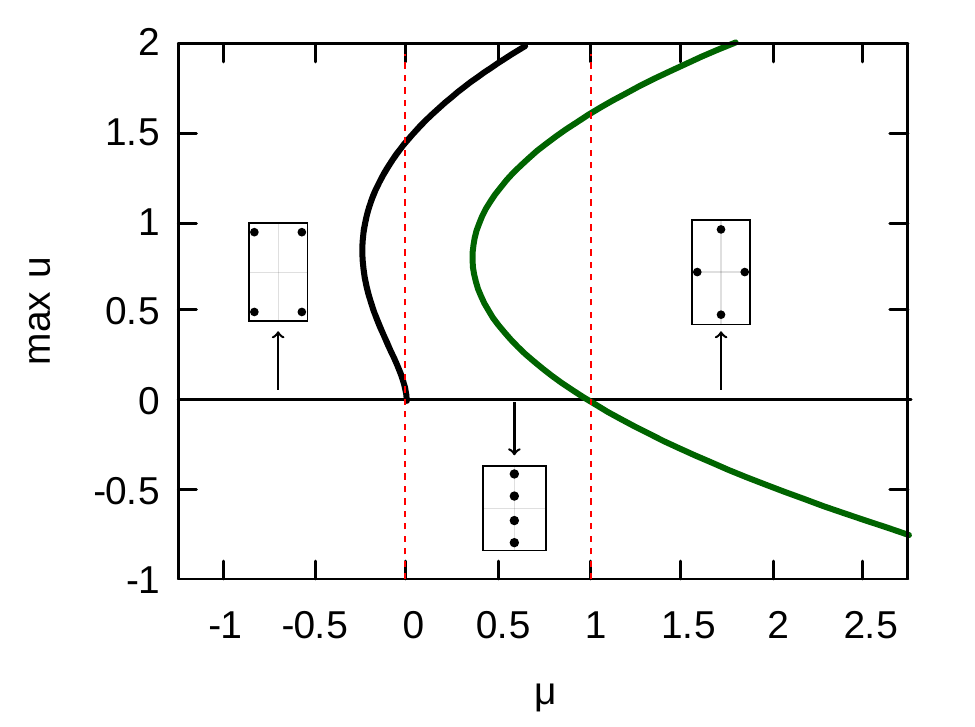} & \includegraphics[width=0.5\linewidth]{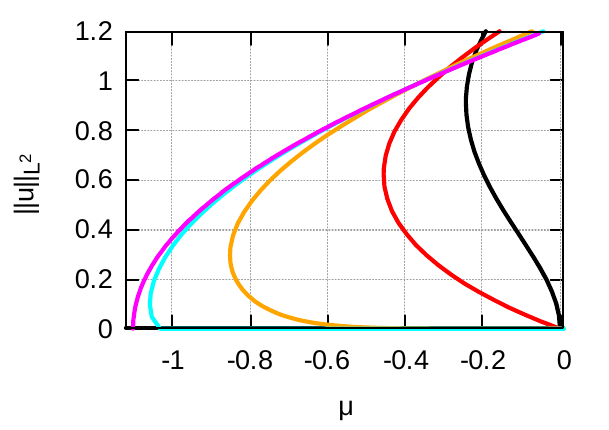} \\
				(a) & (b)
			\end{tabular}
			\caption{ For $\nu = 1.6$ in \eqref{e:spatialNSSHE}: (a) Branches of spatially homogeneous equilibria (green) and $2\pi$-periodic rolls (black) for $\alpha=2$; insets illustrate the spatial eigenvalue configurations at $u=0$ for $\alpha\in(1,2]$. (b) $2\pi$-periodic branches from numerical continuation with \textsc{Auto} \cite{DoedelOldeman2021} for several values of $\alpha$:  $\alpha = 2$ (black), $\alpha = 1.5$ (red), $\alpha = 1.1$ (yellow),  $\alpha = 1.01$ (cyan), and $\alpha =1$ (magenta). For all $\alpha>1$ these branches emanate from $\mu=0$.}
			\label{figure: periodic for different alphas}
		\end{figure}

	\subsection{Bifurcation of rolls from $u=0$ for $\alpha>1$}\label{s:rollsalphabigger1}
		
As a first step to analytically corroborate the numerical results shown in Figure~\ref{figure: periodic for different alphas}(b), we derive an amplitude equation for bifurcating stationary rolls of \eqref{nonsmooth SHE}. These are periodic orbits of \eqref{e:spatialNSSHE} and we first focus on period $2\pi$. 
For this purpose we set $\mu = \sigma \varepsilon^s$, where $\sigma = O(1)$ with respect to $\eps$ and $s$ is to be determined, and seek  solutions to the resulting equation
\begin{equation}\label{pde}
0 = -(1+\partial^2_x)^2 u + \sigma \varepsilon^s u + \nu |u|^{\alpha} -u^3, \; u(x+2\pi)=u(x)
\end{equation}
of the form
\begin{equation}\label{ansatz}
u(x) = \varepsilon A_1\cos(x) + \varepsilon^{1+r} \big(A_0 + A_2\cos(2x)+ E\big), \quad
E = E_\mathrm{high} + O\bigl(\varepsilon^{m}\bigr),
\end{equation}
where $A_j\in\R$ and $r>0$ are to be determined, $m>0$, and $E_{\mathrm{high}}$ denotes terms involving only $\cos(\ell x)$ with $\ell \geq 3$; by spatial translation symmetry we can assume $A_1\geq 0$. 

Inserting the ansatz \eqref{ansatz} into \eqref{pde} and projecting onto cosines, we obtain the amplitude / bifurcation equation, whose expansion allows to determine the nature of the bifurcation. However, standard Taylor expansion cannot be used due to the term $u\mapsto |u|^\alpha$ being $C^1$, but not $C^2$. The following estimate can be used as a substitute for our purposes. 
\begin{proposition} \label{prop:replacement of Taylor}
			Let $v,w\in\R$, $\eps\geq 0$. For any $\alpha \in (1,2]$ we have
			\begin{equation}\label{prop eq 1}
				|v + \eps w|^\alpha = |v|^\alpha + \alpha \eps w\,\mathrm{sgn}(v)
				  |v|^{\alpha-1} + R(v,w,\eps) 
				\;, \quad |R(v,w,\eps)| \leq \alpha \eps^\alpha |w|^\alpha
			\end{equation}
			where  
			$\operatorname{sgn}(\cdot)$ denotes the sign function. 
\end{proposition}
Although here $\eps$ is not necessarily small, it is convenient as a parameter in the proof and we will use the proposition for such values below.
\begin{proof}
First note that \eqref{prop eq 1} holds for $v=0$ or $v+ \eps w=0$ (or $w=0$) so that it remains to consider nonzero $v$ and $v+\eps w$. If $\mathrm{sgn}(v)\neq \mathrm{sgn}(v+\eps w)$ it follows that  $|v|<\eps |w|$ and $|v+\eps w|<\eps |w|$  so that left and right hand side of \eqref{prop eq 1} are indeed both order $\eps^\alpha |w|^\alpha$ as claimed. It thus remains to consider $v,v+\eps w$ of the same sign. 

In case $v,v+\eps w>0$ we have $v+s w>0$ for $s\in(0,\eps)$. Setting $g(\eps):=(v+\eps w)^\alpha$ we have $g'(\eps)=\alpha w (v+\eps w)^{\alpha-1}$ and $g''(\eps) = \alpha(\alpha-1) w^2(v+\eps w)^{\alpha-2}$. 
		By Taylor's theorem
		\[
		g(\eps) = g(0) + \eps g'(0) + \int_0^\eps g''(s) (\eps-s) ds
		= v^\alpha + \eps \alpha w v^{\alpha-1} + R_+(v,w,\eps), 
		\]
		where 
		\[
		R_+(v,w,\eps) = \alpha(\alpha-1)w^2 \int_0^\eps (v+sw)^{\alpha-2}(\eps-s) ds.
		\]
		Since $0\leq \eps-s \leq \eps$ in the integral, due to positivity of the terms we have
		\begin{align}\label{eq:R upper bound}
			\begin{split}
				R_+(v,w,\eps) &\leq \eps \alpha (\alpha-1)w^2 \int_0^\eps (v+sw)^{\alpha-2} ds\\
				&= \eps \alpha w \int_0^\eps \frac{\rmd}{\rmd s}\left((v+sw)^{\alpha-1}\right) ds = \eps \alpha w \left((v+\eps w)^{\alpha-1} - v^{\alpha-1}\right).
			\end{split}
		\end{align}
Since the function \(v \mapsto |v|^{\alpha-1}\) is H\"older continuous of order \(\alpha-1 \in (0,1]\), we have 
\begin{equation}\label{eq: hölder}
\left||v+\eps w|^{\alpha-1} - |v|^{\alpha-1}\right| \leq \tilde C |\varepsilon w|^{\alpha-1},
\end{equation}
for a constant $\tilde C>0$. In fact, we can take $\tilde C=1$ in \eqref{eq: hölder} as shown next in lack of a reference. Replacing $|v+\eps w|$ and $|v|$ by $x$ and $y$, respectively, we need to show $|x^{\alpha-1} - y^{\alpha-1}| \leq |x-y|^{\alpha-1}$. Without loss of generality, assume that $x > y >0$, and define $f(y) = (x-y)^{\alpha-1} - x^{\alpha-1} + y^{\alpha-1}$ for $y \in [0,x]$. We have  $f(0) = f(x) =0$ and from $f'(y) = (\alpha-1)\big(y^{\alpha-2} - (x-y)^{\alpha-2}\big)$ we have (1) $f'(y) =0$ if and only if $y=x/2$, (2) $f'(y) \rightarrow \infty \text{ as } y \rightarrow 0$ (since $\alpha-2<0$) as well as (3) $f'(y) \rightarrow - \infty$ as $y \rightarrow x$. This implies $f(y)>0$ for $0<y \ll 1$ and for $0 < x-y \ll 1$, so that $f$ cannot have zeros in $(0,y)$ since this would require at least two critical points. Therefore $f(y) > 0$ on $(0,x)$ as claimed.

Now, from \eqref{eq: hölder} with $\tilde C=1$ and \eqref {eq:R upper bound} we infer that
\[
R_+(v,w,\eps)  \leq \alpha\eps^\alpha |w|^{\alpha}.
\]	
In the final case $v, v+\eps w<0$ we replace $g$ by $(-v-\eps w)^\alpha$ whose derivative yields a negative factor and the same bound holds for the analogous remainder term. This gives the factor $\mathrm{sgn}(v)$ in the expansion and proves \eqref{prop eq 1}. 
\end{proof}
		
\begin{theorem}\label{thm:rolls1}
Let $\alpha\in(1,2]$. There are $A_*,\mu_*>0$ and an open neighbourhood $U\subset \R^4$ of the origin such that for all $|\mu|<\mu_*$ the following holds. $2\pi$-periodic roll solutions of \eqref{nonsmooth SHE}, whose trajectory lies in $U$ are, up to translation, in 1-to-1 correspondence with solutions $A_1\in[0,A_*)$ to the bifurcation equation 
\begin{equation}\label{eq:genbifeq}
0 = \sigma A_1 + \nu^2 A_1^{2\alpha -1}c_1(\alpha) - \frac{3}{4} \varepsilon^{4-2\alpha} A_1^3 + \calO(\eps^{(\alpha-1)^2}A_1^{\alpha^2}),
\end{equation}
where $\mu = \sigma \eps^s$, $s=2\alpha-2$ and $c_1(\alpha)=\alpha \left(2 c_0(\alpha)^2 + \frac1 9 c_2(\alpha)^2 \right)>0$. These rolls are of the form \eqref{ansatz} with $r=\alpha-1$ and $A_0, A_1, E$ of order $\calO(A_1^\alpha)$. The exact form of $A_0$, $A_2$ is given in \eqref{eq: A0} and \eqref{eq: A2} below.
\end{theorem}

\begin{proof}
The bifurcation equation is obtained from Lyapunov-Schmidt reduction and expansion, e.g. on $L^2([0,2\pi)]$, where the domain of the fourth order symmetric differential operator $-(1+\partial^2_x)^2$ is given by $H^4([0,2\pi])$ and its kernel is spanned by $\cos$ and $\sin$ that are related by translation. This method  requires differentiability of the Nemitsky operator of the nonlinear term, which is standard for the cubic term and for $u\mapsto |u|^\alpha$, $\alpha>1$,  it follows from Theorem 2.7 in \cite{AmbrosettiProdi1993}. This framework justifies the derivation of the bifurcation equation based on substituting \eqref{ansatz} into \eqref{pde} presented next and yields the statement about the relation of rolls and bifurcation equation. 

For the cubic in \eqref{pde}, all terms involving $E$ that results from \eqref{ansatz} are of order $\eps^{r+2}$ or higher and can  be discarded for the leading order analysis. The same holds for the $\calO(\eps^{m)}$-part of $E$ in $|u|^\alpha$. 
Hence, we apply Proposition~\ref{prop:replacement of Taylor} to the leading order part $u(x) = \eps A_1 \cos x + \eps^{r+1}(A_0 + A_2 \cos 2x)$ from \eqref{ansatz}, with $v=A_1 \cos x$, $w=(A_0 + A_2 \cos 2x)$. Then \eqref{prop eq 1} (with $\eps$ replaced by $\eps^r$) gives
		\begin{align}\label{e:ansatzest}
			\begin{split}
				|u(x)|^\alpha =& \varepsilon^\alpha |A_1 \cos x + \varepsilon^r(A_0 + A_2 \cos 2x +E_{\mathrm{high}} )|^\alpha \\
				=&\varepsilon^\alpha |A_1 \cos x|^\alpha + \varepsilon^{\alpha} \mathrm{sgn}( A_1 \cos x ) \alpha \varepsilon^r (A_0 + A_2 \cos 2x +E_{\mathrm{high}} ) |A_1 \cos x|^{\alpha-1} + \calO(\eps^{\alpha(r+1)}) .
			\end{split}
		\end{align}
 
		Substituting $u(x)$ into (\ref{pde}) and using \eqref{e:ansatzest} yields
		\begin{align}\label{u in pde}
			\begin{split}
				0 =& -\varepsilon^{1+r} A_0 - 9\varepsilon^{1+r} A_2 \cos 2x + \sigma \varepsilon^{s+1} A_1 \cos x 
				 + \sigma \varepsilon^{s+1+r} (A_0 + A_2 \cos 2x +E_{\mathrm{high}} ) + \nu \varepsilon^\alpha |A_1 \cos x|^\alpha \\
				& + \nu \varepsilon^{\alpha + r}\mathrm{sgn}( A_1 \cos x ) \alpha (A_0 + A_2 \cos 2x +E_{\mathrm{high}} ) |A_1 \cos x|^{\alpha-1} 
				 -\varepsilon^3 A_1^3 \cos^3x +  \calO(\eps^\rho) 
			\end{split}
		\end{align}
		where  
		$\rho = \min\{ \alpha(r+1), r+3 \}>0$. 
		We 
		next compute the canonical orthogonal projections of \eqref{u in pde} onto the relevant Fourier modes. 
		Projecting onto the constant mode $\cos(0 \, x)=1$,  
		and ignoring the terms in the second line of \eqref{u in pde} for the moment, we obtain
		\begin{equation}\label{eq: 62}
			0 = -\varepsilon^{1+r}A_0 2\pi + \sigma \varepsilon^{s+1+r}A_0 2 \pi + \nu \varepsilon^\alpha \int_{-\pi}^{\pi} |A_1\cos x|^\alpha \,dx, 
		\end{equation} 
		and with $s>0$ this yields to leading order 
		\begin{equation}
			0= -\varepsilon^{1+r}A_0 2\pi + \nu \varepsilon^\alpha A_1^\alpha \int_{-\pi}^{\pi} |\cos x|^\alpha \,dx. 
		\end{equation}
This implies $r=\alpha-1$, which means the leading order in \eqref{u in pde} is $\eps^\alpha$ and the terms in the second line of \eqref{u in pde} are indeed higher order. Also the leading order solution to \eqref{eq: 62} is
		\begin{equation} \label{eq: A0}
			A_0 = \nu A_1^\alpha c_0(\alpha) \;,\quad\text{ where } c_0(\alpha) = \frac{1}{2\pi} \int_{-\pi}^{\pi}|\cos x|^\alpha \,dx >0.
		\end{equation}
		Next, projection of the leading order terms in \eqref{u in pde} onto $\cos (2x)$ yields
		\begin{equation}
			0 = -9 \varepsilon^\alpha A_2 \pi + 
			 \nu \varepsilon^\alpha A_1^\alpha \int_{-\pi}^{\pi} |\cos x|^\alpha \cos 2x \,dx
		\end{equation}
		which gives 
		\begin{equation} \label{eq: A2}
			A_2 = \nu A_1^\alpha \frac{c_2(\alpha)}{9} \;,\quad\text{ where } c_2(\alpha) = \frac{1}{\pi} \int_{-\pi}^{\pi} |\cos x|^\alpha \cos 2x \,dx.
		\end{equation}
		Combining \eqref{eq: A0} and \eqref{eq: A2} we have
		\begin{equation}\label{e:A0A2}
		A_0 +A_2 \cos 2x = \nu A_1^\alpha \left(c_0(\alpha) + \frac{c_2(\alpha)}{9}\cos 2x\right).
		\end{equation}
Finally, projecting \eqref{u in pde} onto $\cos(x)$ and omitting the known higher order terms gives
		\begin{equation} \label{projection on cosx}
			0 = \sigma \varepsilon^{s+1} A_1 \pi - \varepsilon^3  A_1^3 \int_{-\pi}^{\pi}  \cos ^4x \,dx + \nu \int_{-\pi}^{\pi} |\varepsilon A_1 \cos x + \varepsilon^{\alpha} (A_0 + A_2 \cos 2x)|^\alpha \cos x \,dx,
		\end{equation}
		where the first integral on the right hand side is $\frac 3 4 \pi$. For the second integral using \eqref{prop eq 1} 
		yields
		\begin{align}\label{prop integral}
			\begin{split}	
				&\int_{-\pi}^{\pi} |A_1 \cos x + \varepsilon^{\alpha-1} (A_0 + A_2 \cos 2x)|^\alpha \cos x \,dx = \\
				& \qquad A_1^\alpha \int_{-\pi}^{\pi} |\cos x|^\alpha \cos x \,dx 
				+ \varepsilon^{\alpha-1} \alpha A_1^{\alpha-1} \int_{-\pi}^{\pi}  (A_0 + A_2 \cos 2x) |\cos x|^{\alpha} \,dx 
+ \calO(\varepsilon^{\alpha(\alpha-1)}A_1^{\alpha^2}),
			\end{split}
		\end{align}
where $A_1^{\alpha^2}$ in the higher order term is  
due to \eqref{e:A0A2}. In \eqref{prop integral} the first integral on the right hand side vanishes since the $2\pi$-periodic integrand is an odd function with respect to $x=\pi/2$.  		
		Inserting \eqref{prop integral} into \eqref{projection on cosx} and using \eqref{e:A0A2} we find
		\begin{align}
			\begin{split}\label{last equation}
				&0 = \sigma \varepsilon^{s+1} A_1 \pi - \varepsilon^3 A_1^3 \frac{3\pi}{4} 
				+ \nu^2 \varepsilon^{2\alpha-1} \alpha A_1^{2\alpha -1} \int_{-\pi}^{\pi} \biggl(c_0(\alpha) + \frac{c_2(\alpha)}{9} \cos 2x \biggl) |\cos x|^\alpha \,dx + \calO(\varepsilon^{\alpha^2}A_1^{\alpha^2}). 
			\end{split}
		\end{align} 
		Balancing terms of lowest order implies the scaling 
		\begin{equation}
			s = 2\alpha -2.
		\end{equation}		
		Hence, we obtain from \eqref{last equation} the general bifurcation equation \eqref{eq:genbifeq} as claimed, 
		where, using \eqref{eq: A0} and \eqref{eq: A2}, 
		\begin{equation} \label{c_1}
			c_1(\alpha) := \frac{\alpha}{\pi} \int_{-\pi}^{\pi} \biggl( c_0(\alpha) + \frac{c_2(\alpha)}{9} \cos 2x \biggl) |\cos x|^\alpha \,dx 
			 =  \alpha \left(2 c_0(\alpha)^2 + \frac1 9 c_2(\alpha)^2 \right). 
		\end{equation}
In particular, $c_1(\alpha)>0$ for all $\alpha\in(1,2)$ since $c_0(\alpha)>0$. (In fact, also $c_2(\alpha)>0$. See Lemma \ref{lemma c_i's}.)  

 Lastly, as a rough estimate of $E$ in terms of $A_1$ we note that \eqref{prop eq 1} analogous to the estimates above implies $E$ is at least order $A_1^\alpha$, just as $A_0, A_2$; also recall that without loss $A_1\geq 0$ in \eqref{ansatz} due to translation in $x$.
\end{proof}

\begin{figure}
	\centering
	\includegraphics[width=0.7\linewidth]{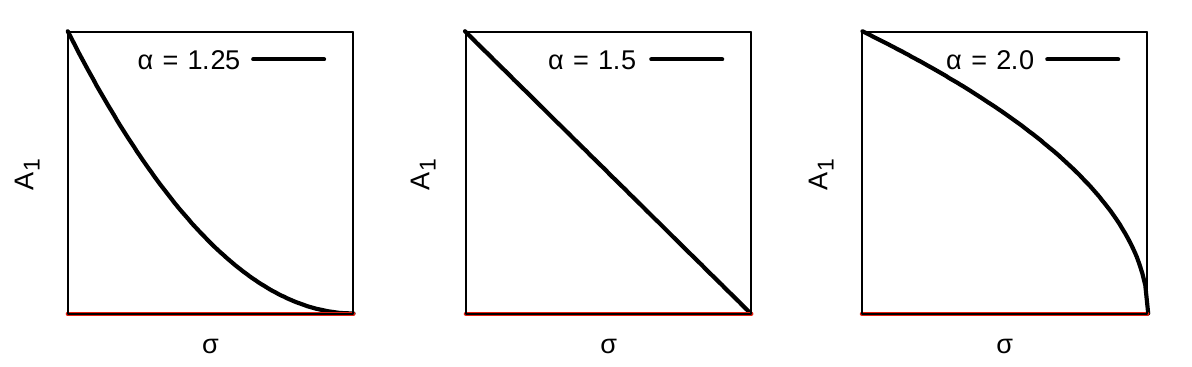}
	\caption{ Bifurcation diagrams of the bifurcation equation \eqref{eq:genbifeq} for $\alpha = 1.25$, $1.5$ and $2$ (subcritical).}
	\label{figure: A vs sigma}
\end{figure}

In particular, for $\alpha \in (1,2)$, the bifurcation equation \eqref{eq:genbifeq} can be written as  
\begin{equation}\label{eq:bifeq_al1}
0 = \sigma A_1 +\nu^2c_1(\alpha) A_1^{2\alpha -1} + O\big(\varepsilon^{{\rho}} A_1^{\tilde\rho} \big), \; {\rho} = \min \big\{ (\alpha-1)^2, 4-2\alpha \big\}>0,\;  \tilde\rho =  \min\{3,\alpha^2\}>2\alpha-1>1 .
\end{equation} 

\begin{corollary}\label{cor:rollsubcrit}
For $\alpha \in (1,2)$, the bifurcation of the $2\pi$-periodic rolls from the trivial solutions of \eqref{nonsmooth SHE} is subcritical for all $\nu\neq0$. 
\end{corollary}

\begin{proof}
By Theorem~\ref{thm:rolls1}, for $\alpha\in(1,2)$ the non-trivial solutions to \eqref{eq:bifeq_al1} are to leading order 
\begin{equation}\label{eq:rollsigma}
A_1 = A(\sigma):=\biggl( \sqrt{ -\frac{\sigma}{\nu^2 c_1(\alpha)}} \biggr) ^ {\frac{1}{\alpha-1}},
\end{equation}
where $c_1(\alpha) >0$ for $\alpha \in (1,2)$. Since denominator is positive for $\nu \neq 0$ real solutions require $\sigma < 0$. Hence, the bifurcation of $2\pi$ periodic roll solutions from $\mu=0$ is always subcritical. 
\end{proof}

 This corollary is significant since it highlights the non-trivial impact of $\alpha\in(1,2)$, i.e., the non-smooth situation, on the criticality:  
For the smooth case $\alpha = 2$, i.e., nonlinearity $f(u) = \nu u^2 - u^3$, we have $c_1(2) = 19/18$ due to $c_0(2) = 1/2$ and $c_2(2) = 1/2$ in \eqref{c_1} so that the bifurcation equation gives the classical 
\begin{equation}\label{e:bifeqal2}
0 = \sigma A_1 + \nu^2 \frac{19}{18} A_1^3 - \frac{3}{4} A_1^3 + \calO(\varepsilon A_1^4 ) 
= A_1 \biggl( \sigma - \biggl( \frac{3}{4} - \nu^2 \frac{19}{18} \biggl) A_1^2 \biggl) + \calO(\varepsilon A_1^4).
\end{equation} 
Hence, the bifurcation in terms of $\sigma$ is a subcritical pitchfork if  $\frac{3}{4} - \nu^2 \frac{19}{18} < 0$, i.e., $\nu^2 >\frac{27}{38}$ and is supercritical if $\nu^2 <\frac{27}{38}$. In other words, for $\alpha=2$, when increasing $\nu$ from zero, the bifurcation is initially supercritical and switches to subcritical only beyond a threshold value. (Recall that we consider $\nu\geq 0$ without loss.) In contrast, for $1<\alpha<2$ the bifurcation is subcritical \textit{for any} $\nu\neq 0$, as shown in Corollary \ref{cor:rollsubcrit}. We note that the value $\nu=1.6$, that we use in most computations, is larger than the threshold for $\alpha=2$ and thus gives a subcritical bifurcation in that case.  
We plot such leading order branches near the origin for different values of $\alpha$ in Figure \ref{figure: A vs sigma}. For $\alpha > 3/2$ the branch is concave, for $\alpha =3/2$ it is linear, and for $\alpha< 3/2$ it is convex, leading to an increasingly flat branch as $\alpha$ approaches $1$.

In Figure \ref{figure: fold continuation} we further illustrate this in terms of $\mu=\sigma \eps^{2\alpha-2}$ and the $L^2$-norm. In particular, the $2\pi$-periodic branch bifurcates to the left and is increasingly flat, reflecting the stronger subcriticality in this regime. For the interpretation of \eqref{eq:bifeq_al1} in terms of the parameter $\mu$, we multiply \eqref{eq:bifeq_al1} by $\varepsilon ^{2\alpha -1}$, which yields
\begin{equation}\label{eq:bifeq_al1_1}
0 = \mu \cdot(\varepsilon A_1) +\nu^2 c_1(\alpha) \cdot(\varepsilon A_1)^{2\alpha -1} + O((\varepsilon A_1)^{\tilde\rho}),
\end{equation}
whose nontrivial leading order solution is given by $\varepsilon A_1 = A(\mu)$ from \eqref{eq:rollsigma}. Since $u(x) \approx \varepsilon A_1 \cos (x)$, 
 \[
\|u\|_{L^2(-\pi,\pi)} \approx \varepsilon A_1 \sqrt{\pi} =  A(\mu) \sqrt{\pi}.
\]

Beyond the local bifurcation analysis, the numerical continuation results plotted in Figure~\ref{figure: fold continuation} show that for each $\alpha\in(1,2)$ the branch features one fold point at some $\mu^*(\alpha,\nu)<0$ and connects back to $\mu=0$ at a roll with non-zero amplitude so that for $\mu\in(\mu^*(\alpha,\nu),0)$ two non-zero roll solutions coexist. Moreover, the numerical results suggest that the fold point tends to some $\mu^*(\nu)$ at $u=0$ as $\alpha\searrow 1$. It thus appears that at $\alpha=1$ there is no longer a region of coexisting non-zero roll solutions, but only a single branch of non-zero rolls remains. Furthermore, this remaining branch appears to bifurcate supercritically from $\mu= \mu^*(\nu)$, where, e.g., $\mu^*(1.6) \approx -1.10$ 
and $\mu^*(0.5) \approx -0.18$. 

In addition, the numerical results for $\nu = 0.5$ shown in Figure~\ref{figure: fold continuation}(b) corroborate the absence of a subcriticality threshold for $\alpha\in(0,2)$ as discussed above: although the branch is supercritical for $\alpha = 2$, it is subcritical for all $\alpha \in (1,2)$. 

		\begin{figure}
			\centering
			\begin{tabular}{cc}
			\includegraphics[width=0.45\textwidth]{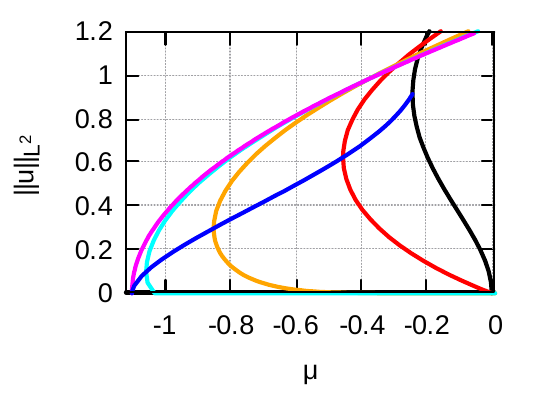} &
			\includegraphics[width=0.45\textwidth]{images/nu_0p5.pdf}	\\
			(a) & (b)
			\end{tabular}
\caption{Branches of  $2\pi$ periodic rolls in \eqref{e:spatialNSSHE} at $\alpha = 2$ (black) , $\alpha = 1.5$ (red), $\alpha = 1.1$ (orange), $\alpha = 1.01$ (cyan), $\alpha =1$ (magenta) for (a) $\nu = 1.6$, together with the locations of fold points (blue), and (b) $\nu = 0.5$ .}
			\label{figure: fold continuation}
		\end{figure}
		
\bigskip
We turn to another analytical aspect of the bifurcation from $\mu=0$ for $\alpha\in(1,2)$, namely the next order correction in \eqref{ansatz} to the pure cosine form of the bifurcating rolls. This is given by $A_0+ A_2\cos(2x)$, where $A_0, A_2$ from \eqref{eq: A0}, \eqref{eq: A2} contain the factors $c_0, c_2$, respectively, and next provide explicit formulae for these.
		
		\begin{lemma} \label{lemma c_i's}
			For $\alpha \in (1,2]$, the functions $c_0(\alpha)$, $c_2(\alpha)$, and $c_1(\alpha)$, defined in \eqref{eq: A0}, \eqref{eq: A2}, and \eqref{c_1} respectively, can be computed explicitly as the positive values
			\begin{equation*}
				c_0(\alpha) = \frac{1}{\sqrt{\pi}} \frac{\Gamma(\frac{1+\alpha}{2})}{\Gamma (1+\frac{\alpha}{2})}, \;
				c_2(\alpha) =  \frac{\alpha}{\sqrt{\pi}} \frac{\Gamma (\frac{1+\alpha}{2})}{\Gamma(2+\frac{\alpha}{2})}, \; 
				c_1(\alpha) = \frac{\alpha}{\pi} \frac{(36+36\alpha + 11\alpha^2)}{18} \frac{\Gamma^2(\frac{1+\alpha}{2})}{\Gamma^2 (2+\frac{\alpha}{2})} .
			\end{equation*}
		\end{lemma}
		
		\begin{proof}		
		We begin by evaluating the integral in \eqref{eq: A0},  
		which can be written as
		\begin{equation*}
			\int_{-\pi}^{\pi}|\cos(x)|^{\alpha} \,dx = 4 \int_{0}^{\pi/2} (\cos x)^{\alpha} \,dx
		\end{equation*}
		and, by using change of variables $u=\sin^2x$ and $du = 2 \sin x \cos x \,dx$, we get
		\begin{equation*}
			2 \int_{0}^{1} (1-u)^{\alpha/2} (u(1-u))^{-1/2} \,du = 2 \int_{0}^{1} (1-u)^{\frac{\alpha-1}{2}}u^{-\frac{1}{2}} \,du .
		\end{equation*}
		 Using the Beta function \cite{andrews1999special}  $\text{B}(m,n) = \frac{\Gamma (m) \Gamma(n)}{\Gamma(m+n)} = \int_{0}^{1} u^{m-1}(1-u)^{n-1} \,du $ for all $m,n >0$, where $\text{B}(m,n) = \text{B}(n,m)$, from $\Gamma(\frac{1}{2}) = \sqrt{\pi}$ 
we obtain 
		\begin{equation}
			\int_{-\pi}^{\pi} |\cos(x)|^\alpha \,dx =
			2 \text{B} \biggl(\frac{1}{2}, \frac{\alpha+1}{2} \biggr) = 2 \frac{\Gamma (\frac{1}{2}) \Gamma (\frac{\alpha+1}{2})}{\Gamma(z+1)} = 
			\frac{2 \sqrt{\pi} \Gamma(\frac{1+\alpha}{2})}{\Gamma(1+\frac{\alpha}{2})}
		\end{equation} 
		and thus $c_0(\alpha)$ as claimed. 
		Next, we evaluate the integral in \eqref{eq: A2}, which can be written as 
		\begin{equation}\label{eq: int}
			\int_{-\pi}^{\pi} |\cos x|^{\alpha} \cos(2x) \,dx = 2\int_{-\pi}^{\pi} |\cos x|^{\alpha} \cos^2x \,dx - \int_{-\pi}^{\pi} |\cos x|^{\alpha} \,dx .
		\end{equation} 
		From above, the second integral on right hand side is $2\text{B} \bigl(\frac{1}{2}, \frac{\alpha+1}{2} \bigr)$. The first integral on right hand side of \eqref{eq: int} can be written as
		\begin{equation*}
			2\int_{-\pi}^{\pi} |\cos x|^{\alpha} \cos^2x \,dx = 2\int_{-\pi}^{\pi} |\cos x|^{2(\alpha/2+1)} \,dx = 8\int_{0}^{\frac{\pi}{2}} (\cos x)^{2(\alpha/2+1)} \,dx
		\end{equation*}		
		and the change of variables $u=\sin^2x$ and $du = 2\sin x \cos x \,dx $ gives
		\begin{equation}
			4 \int_{0}^{1} (1-u)^{\frac{\alpha+1}{2}} (u(1-u))^{-\frac{1}{2}} \,du = 4 \int_{0}^{1} (1-u)^{\frac{\alpha+1}{2}} u^{-\frac{1}{2}} \,du = 4 \text{B} \biggl(\frac{1}{2}, \frac{\alpha+3}{2} \biggr) .
		\end{equation}
		Hence, the right hand side of \eqref{eq: int} becomes, with $z=\alpha/2$, 
		\begin{equation}
			4 \text{B} \biggl(\frac{1}{2}, z+ \frac{3}{2} \biggr) - 2\text{B} \biggl(\frac{1}{2}, z+ \frac{1}{2} \biggr) = 4\frac{\Gamma(\frac{1}{2}) \Gamma(z+\frac{3}{2})}{\Gamma(z+2)} - 2 \frac{\Gamma (\frac{1}{2}) \Gamma (z+ \frac{1}{2})}{\Gamma (z+1)}.
		\end{equation}
		From $\Gamma(z+2) = (z+1) \Gamma(z+1)$ and $\Gamma (\frac{1}{2}) = \sqrt{\pi}$ the right hand side can be written as  
		\begin{equation}
			\frac{2\sqrt{\pi}}{\Gamma\big(z+2\big)}  \left( 2\Gamma\left(z+\frac{3}{2}\right) - (z+1)\Gamma\left(z+\frac{1}{2}\right) \right), 
		\end{equation}
		and using $\Gamma\left(z+\frac{3}{2}\right) = (z+\frac{1}{2}) \Gamma(z+\frac{1}{2})$, the expression in the parentheses becomes
		\begin{equation} 
			\Gamma\left(z+\frac{1}{2}\right) \left( 2(z+\tfrac{1}{2})- (z+1) \right) = \Gamma\left(z+\frac{1}{2}\right)z.
		\end{equation}
		Hence, \eqref{eq: int} takes the value 
		\begin{equation}
			\int_{-\pi}^{\pi} |\cos x|^\alpha \cos(2x) \,dx = 2z\sqrt{\pi}\frac{ \Gamma(z+\frac{1}{2})}{\Gamma(z+2)} = \alpha \sqrt{\pi}\frac{\Gamma(\frac{1+\alpha}{2})}{\Gamma(2+\frac{\alpha}{2})},
		\end{equation}			
which gives $c_2(\alpha)$ as claimed. 

Finally, substituting the formulae for $c_0$ and $c_2$ into \eqref{c_1}  gives 
		\begin{equation*}
			c_1(\alpha)  =2 \alpha c_0^2(\alpha) + \frac{1}{9} \alpha c_2^2 (\alpha)
= \frac{\alpha}{\pi} \left( 2\frac{\Gamma^2 (\frac{1+\alpha}{2})}{\Gamma^2 (1+\frac{\alpha}{2})} + \frac{\alpha^2}{9} \frac{\Gamma^2 (\frac{1+\alpha}{2})}{\Gamma^2(2+\frac{\alpha}{2})} \right) 
				 = \frac{\alpha}{\pi} \frac{(36+36\alpha + 11\alpha^2)}{18} \frac{\Gamma^2(\frac{1+\alpha}{2})}{\Gamma^2 (2+\frac{\alpha}{2})},
		\end{equation*}
		where the third equality follows from $\Gamma(2+\frac{\alpha}{2}) = (1+\frac{\alpha}{2}) \Gamma(1+\frac{\alpha}{2})$. 		
		\end{proof}

\subsection{Bifurcation of periodic rolls at $\alpha =1$}\label{subsec: Bifurcation of periodic rolls}
	
Due to the lack of differentiability at $\alpha=1$, 
our goal is to identify values of $\mu$ at which bifurcations of rolls from $u\equiv 0$ occur. We first note that non-trivial long term dynamics is impossible for $\mu\leq -\nu$. (Recall we consider $\nu\geq 0$ without loss.)

\begin{lemma}\label{lem:globaldecay}
Consider the Swift-Hohenberg equation \eqref{nonsmooth SHE} for $\alpha=1$, i.e., 
\begin{equation} \label{SH3x}
\partial_t u = - (1+\partial_x^2)^2 u + \mu u + \nu |u| -u^3
\end{equation}
posed on $\R$ or on an interval with periodic boundary conditions. 
Then, for $\mu\leq -\nu$ all $L^2$-solutions decay to zero. For $\mu< -\nu$ this is exponential with $\|u(t)\|_2 \leq \rme^{(\mu+\nu)t}\|u(0)\|_2$.  
\end{lemma}

\begin{proof}
	Multiplying \eqref{SH3x} with $u$ and integrating over the domain gives 
	\begin{align}\label{e:L2decay}
		\begin{split}
			\frac{1}{2}\frac{\rmd}{\rmd t}\|u\|_2^2 &= -\int u (1+\partial_x^2)^2 u \rmd x + \mu \|u\|_2^2 + \nu \int |u|u\rmd x - \|u^2\|_2^2 \\
			& \leq -\|(1+\partial^2_x)u\|_2^2 + (\mu+\nu) \|u\|_2^2 - \|u^2\|_2^2 \ \leq (\mu+\nu) \|u\|_2^2,	
		\end{split}
	\end{align}
	where we used that $(1+\partial_x^2)$ is symmetric and $\langle |u|,u \rangle_2 \leq \|u\|_2^2$. 
	Hence, $\|u(t)\|_2^2 \leq \rme^{2(\mu+\nu)t}\|u(0)\|_2^2$, which proves the claim for $\mu<-\nu$. The claim for $\mu=-\nu$ follows from noting that the last inequality in \eqref{e:L2decay} is strict whenever $\|u\|_2>0$. 
\end{proof}

\begin{remark}
It is straightforward to derive an analogous estimate for \eqref{nonsmooth SHE} and any $\alpha\in[1,2]$ so that energy decays for $\mu\leq -\tilde\nu_\alpha$ for suitable $\tilde\nu_\alpha>0$, e.g. $\tilde\nu_1=\nu$ due to Lemma~\ref{lem:globaldecay}. For the general case one uses that $\nu |u|^\alpha u - u^4\leq \tilde\nu_\alpha u^2$ so that, for instance, $\tilde\nu_2=\tfrac{\nu^2}{4}$ can be chosen. 
\end{remark}

 \begin{remark}\label{rem:energynondecay}
 On the other hand, energy cannot decay to zero in this generality in the presence of rolls or instability of $u=0$. For $\alpha\in(1,2]$ the latter is the case whenever $\mu>0$, as can be readily determined from the spectrum. At $\mu=0$ this is more subtle, but occurs whenever the bifurcation to rolls is subcritical. At $\alpha=1$ there is no linearisation, but $u=0$ becomes unstable for some $\mu\in(-\nu,0]$. 
 Indeed, consider an initial datum $u_0(x)=a\cos(x)$ with $0<a\ll1$ and let $\hat u_1(t)$ denote the Fourier coefficient of $\cos$ for the associated solution. Since $|\cos(x)|\cos(x)$ has zero mean, we have $\tfrac{d}{dt}\hat u_1(0) \approx \mu a$, which is positive whenever $\mu>0$ (for all sufficiently small $a>0$) so that energy does not decay. In fact, this proves instability of $u=0$ for $\mu>0$ and thus any primary bifurcation occurs for some $\mu\in (-\nu,0]$.
 \end{remark}

Based on these observations, for $\alpha=1$ we expect a bifurcation to rolls at some value $\mu\in(-\nu,0]$. 
We recall the numerical observation that the fold points $\mu^*(\alpha,\nu)$ appear to limit, as $\alpha\searrow 1$, to a bifurcation point $\mu^*(\nu)$ of $2\pi$-periodic rolls for $\alpha = 1$. Lemma~\ref{lem:globaldecay} implies $\mu^*(\nu)>-\nu$ and since the subcriticality implies $\mu^*(\alpha,\nu)<0$, it is natural to expect $\mu^*(\nu)<0$. Numerical continuation of the folds confirms this as shown in Figure~\ref{figure: fold continuation}(a). However, it seems difficult to prove this or just the uniqueness of a bifurcation point due to the lack of differentiability. 

\medskip
For an analytic treatment to further corroborate the numerical results, we next characterize sign-changing rolls by an appropriate boundary value problem. First, we reduce to consider the leading order form of the Swift-Hohenberg equation \eqref{nonsmooth SHE} for small amplitude steady states by ignoring the cubic term, which does not change $\mu^*(\nu)$. Hence, for  $\alpha=1$, we consider		
		\begin{equation} \label{SH1}
			0 = - (1+\partial_x^2)^2 u + \mu u + \nu |u|, 
		\end{equation}
which features the Lipschitz nonlinearity $|u|$ and is linear for $u>0$ and for $u<0$ with switching set $\{u=0\}$. Equation \eqref{SH1} is positively homogeneous: if $u$ is a solution, then $a u$ is also a solution for any $a \geq 0$, i.e., solutions come in cones. Just as \eqref{nonsmooth SHE}, \eqref{SH1} is symmetric with respect to spatial reflection $x \mapsto -x$.

Since the absolute value function is not differentiable at zero, solutions to \eqref{SH1} (as a 4-th order equation), which cross the switching set, are generally $C^4$-smooth, but not $C^5$. 
Recall from \S\ref{sec:homostatesevals} the roll solutions accompanying the homogeneous steady states $u_{\pm 0}$ that fully lie in $\{u>0\}$ or $\{u<0\}$, respectively. Here $u_{\pm 0}$ bifurcated supercritically at half-pitchforks from $u=0$ at $\mu=1\pm\nu$, cf.\ \S\ref{sec:homostates}, as shown in Figure~\ref{figure: homogeneous for alpha 1}. For \eqref{SH1} these half-pitchforks turn into vertical branches and here we seek sign changing periodic solutions that bifurcate from the trivial state. We expect that a branch of such solutions yields a branch to \eqref{nonsmooth SHE} that is perturbed, e.g., using the results in \cite{KuepperHosham2011}.

Due to the reflection symmetry $x\to -x$ of \eqref{SH1}, we seek solutions $u(x)$ to \eqref{SH1} posed on an interval $[0,L]$ that can be extended to solutions on $[-L,L]$ by reflection, $u(x)=u(-x)$, and thus to periodic solutions on $\R$.  Hence, we consider the analogue of homogeneous Neumann boundary conditions where odd derivatives vanish, i.e. $u'(x) = u'''(x)=0$ at $x=0,L$.  
For such solutions of \eqref{SH1} that change sign once, it is convenient to shift the location of sign change from negative to positive to $x=0$ and pose \eqref{SH1} on an interval $[-L_-,L_+]$ with these boundary conditions and $L=L_++L_-$. We thus seek solutions of the form
		\begin{equation}\label{e:bvp1}
			u(x) = \begin{cases}
				u_+(x), & x \in (0,L_+]\\
				0, &x=0,\\
				u_-(x), & x \in [-L_-, 0)
			\end{cases}
		\end{equation}
where $u_+>0>u_-$ and $L_++L_-$ is the half-period of $u$ as a solution posed on $\R$.
See Figure \ref{figure: periodic solution image} for illustration and note the (at least) $C^4$-smoothness at the root of $u$. 		
		\begin{figure}
			\centering
			\includegraphics[width=0.35\linewidth]{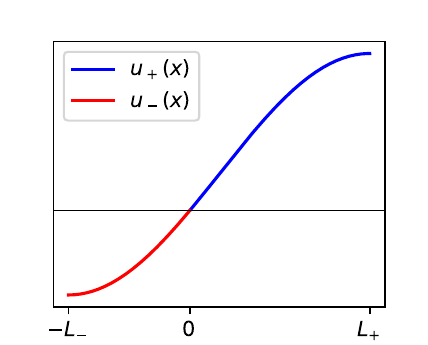}	
			\caption{Illustration of a possible (at least) $C^4$-smooth solution as in Lemma \ref{lem:symperBVP}. 
			}
			\label{figure: periodic solution image}
		\end{figure} 
		Concerning boundary conditions, as mentioned before, for reflection symmetric solutions we have 
		\begin{equation}\label{neumann conditions}
			u_{-}^{(i)} (-L_{-})=0 \text{ and } u_{+}^{(i)} (L_{+})=0 \text{ for } i=1,3,
		\end{equation}
		together with the matching conditions for continuity
		\begin{equation}\label{matching condition}
			u_{+} (0)=u_{-} (0) =0
		\end{equation}
		and smoothness 
		\begin{equation}\label{smoothness condition}
			u_{+}^{(j)} (0)=u_{-}^{(j)} (0) \text{ for } j=1,2,3 .           	
		\end{equation}
		
Equation \eqref{SH1} being linear for fixed sign of $u$, $\nu |u| = \nu u_+$ and $\nu |u| = -\nu u_-$ yields the general solutions  
		\begin{equation}\label{eq:u+-sol}
			u_\pm(x) = a_{\pm 1} e^{i \kappa_{\pm 1}x} +  a_{\pm 2} e^{-i \kappa_{\pm 1}x} +  a_{\pm 3} e^{i \kappa_{\pm 2}x} +  a_{\pm 4} e^{-i \kappa_{\pm 2}x}
		\end{equation}
		where $a_{\pm j}\in\C$. The complex frequencies $\kappa_{\pm j}$ result from the $(\nu,\mu)$-dependent solutions to the characteristic equations of \eqref{SH1} for $u>0$ and $u<0$, respectively, and are given by 
		\begin{equation}\label{eigenvalues}
			\begin{aligned}
				\kappa_{+1} = \sqrt{1 + \sqrt{\mu + \nu}} &,& \kappa_{-1} = \sqrt{1 + \rmi \sqrt{\nu - \mu}}\\
				\kappa_{+2} = \sqrt{1 - \sqrt{\mu + \nu}}  &,& \kappa_{-2} = \sqrt{1 -\rmi \sqrt{\nu - \mu}}.	
			\end{aligned}
		\end{equation} 
Figure \ref{figure: spectral plot} illustrates the loci of these in different regions of $\mu$ defined by changes in the real parts for $u>0$ and/or $u<0$. There are two different scenarios depending on whether $\nu>\frac 1 2$ or $\nu<\frac 1 2$. 

Recall from \S\ref{sec:homostatesevals} that in the smooth case the purely imaginary spatial eigenvalues determine (up to the resonant case) the existence of rolls near $u=0$ via the Lyapunov centre theorem.  We are not aware of any similar generally applicable abstract result for the present case, where the spatial eigenvalues switch. Nevertheless, the character of the spatial eigenvalues is clearly relevant for the boundary value problem corresponding to such solutions and we note the results reviewed in \cite{KuepperHoshamWeiss2013}. This relies on determining the eigenvalue of a Poincar\'e map, which, as far as we can see, is no simpler that solving the boundary value problem. Results of this form we are aware of pertain 2D or 3D, with the exception \cite{Hosham2018} that, however, does not overlap with our setting.  

By construction and uniqueness of solutions to \eqref{SH1} for given $\partial_x^ju(0)$, $j=0,1,2,3$, we have the following. 
\begin{lemma}\label{lem:symperBVP}
For each given $\nu,\mu$, the reflection symmetric $2L$-periodic solutions of \eqref{SH1} with two sign changes over one period are, via  \eqref{e:bvp1}, \eqref{eq:u+-sol}, \eqref{eigenvalues}, in 1-to-1 correspondence with $L_\pm>0$ and $a_{\pm j}\in\C$, $j=1,\ldots,4$ such that $L_++L_-=L$, $u_\pm$ satisfy \eqref{neumann conditions}-\eqref{smoothness condition} and $u$ in \eqref{e:bvp1} is real valued. 
\end{lemma}		

		\begin{figure}
			\begin{tabular}{cc}
			\includegraphics[width=0.45\linewidth]{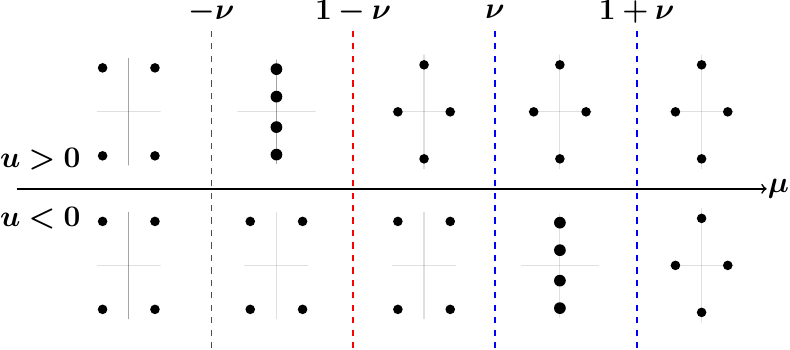}& 
			\includegraphics[width=0.45\linewidth]{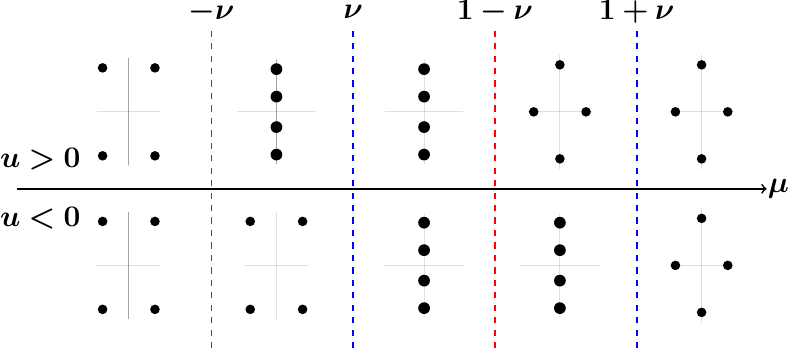}\\
			(a) & (b)
			\end{tabular}
			\caption{Plot of the qualitative arrangements of the spatial eigenvalues $\rmi\kappa_{\pm j}$, $j=1,2$ from \eqref{eigenvalues}  for $\alpha = 1$ for (a) $\nu > 1/2$ and (b) $\nu < 1/2$.}
			\label{figure: spectral plot}
		\end{figure} 
		
In order to analyze bifurcations for $\mu>-\nu$, which is the relevant range due to Lemma~\ref{lem:globaldecay}, we focus on the $2\pi$-periodic case so that $L_+ + L_- = \pi$ and first consider $\nu > 1/2$, $\mu \in (-\nu,1-\nu]$. Then $\kappa_{+1}, \kappa_{+2}\in\R$, $\mathrm{Re}(i \kappa_{-2}) = -\mathrm{Re} (i \kappa_{-1})$ and $\mathrm{Im}(i \kappa_{-2}) = \mathrm{Im} (i \kappa_{-1})$. We next show that this admits a reduction of the boundary value problem \eqref{e:bvp1}-\eqref{smoothness condition} to a system of algebraic equations involving the quantities 
	\begin{align} \label{eq: gamma's}
		\begin{split}
		\eta = & \mathrm{Re}(i \kappa_{-2}), \quad \beta = \mathrm{Im}(i \kappa_{-2}),\\
			\gamma_{31} =& (-1 + e^{2 \eta (L_+ -\pi)}) (\kappa_{+1}^2 - \kappa_{+2}^2)-2 \eta (1 + e^{2 \eta (L_+ -\pi)}) \kappa_{+1} \tan(\kappa_{+1} L_+) \\
			&  +2 \eta (1 + e^{2 \eta (L_+ -\pi)}) \kappa_{+2} \tan(\kappa_{+2} L_+), \\
			\gamma_{32} =& \eta (1 + e^{2 \eta (L_+ -\pi)}) (\kappa_{+1}^2 - \kappa_{+2}^2) + (-1 + e^{2 \eta (L_+ -\pi)}) \kappa_{+1} (\eta^2 - \beta^2 + \kappa_{+1}^2) \tan(\kappa_{+1} L_+) \\
			& - (-1 + e^{2 \eta (L_+ -\pi)}) \kappa_{+2} (\eta^2 - \beta^2 + \kappa_{+2}^2) \tan(\kappa_{+2} L_+), \\
			\gamma_{41} =& (3 \eta^2 - \beta^2)(-1 + e^{2 \eta (L_+ -\pi)})(\kappa_{+1}^2 - \kappa_{+2}^2) +2 \eta (1 + e^{2 \eta (L_+ -\pi)}) \kappa_{+1}^3 \tan(\kappa_{+1} L_+) \\
			& -2 \eta (1 + e^{2 \eta (L_+ -\pi)}) \kappa_{+2}^3 \tan(\kappa_{+2} L_+), \\
			\gamma_{42} =& \eta (\eta^2 -3\beta^2) (1 + e^{2 \eta (L_+ -\pi)}) (\kappa_{+1}^2 - \kappa_{+2}^2) \\
			& + (-1 + e^{2 \eta (L_+ -\pi)}) \kappa_{+1} \big( \eta^4 + \beta^4 -\beta^2 \kappa_{+1}^2 + \eta^2 (2\beta^2 + \kappa_{+1}^2) \big) \tan(\kappa_{+1} L_+) \\
			& - (-1 + e^{2 \eta (L_+ -\pi)}) \kappa_{+2} \big( \eta^4 + \beta^4 -\beta^2 \kappa_{+2}^2 + \eta^2 (2\beta^2 + \kappa_{+2}^2) \big) \tan(\kappa_{+2} L_+).
		\end{split}
	\end{align} 
In particular all these depend on $\mu,\nu$. 

	\begin{theorem}\label{theorem: reduced system nu>1/2}
For given $\nu > 1/2$, the reflection symmetric $2\pi$-periodic real valued solutions of \eqref{SH1} with two sign changes over one period with $\mu\in (-\nu,1-\nu)$  are in 1-to-1 correspondence with solutions $(\mu,L_+)\in (-\nu,1-\nu)\times [0,\pi]$ to 
		\begin{align}
			\begin{split}\label{reduced system}
				&f_\nu (\mu,L_+) := \gamma_{31} \gamma_{42} - \gamma_{41} \gamma_{32} = 0 \\
				&g_\nu (\mu,L_+) := \beta \cos\big(\beta (L_+ -\pi)\big) \gamma_{31} + \sin\big(\beta (L_+ - \pi)\big) \gamma_{32} = 0
			\end{split}
		\end{align}
		with $\gamma_{ij}, \beta, \eta$ defined in \eqref{eq: gamma's}.
\end{theorem}

		\begin{proof} 
		For $-\nu < \mu < 1-\nu$ and $\nu>1/2$ the solution $u_+$ can be written as
		\begin{align}
			\begin{split}
				u_+(x) =& (a_{+1} + a_{+2}) \cos (\kappa_{+1} x) + i ( a_{+1} - a_{+2} ) \sin (\kappa_{+1} x) \\
				&+ ( a_{+3} + a_{+4}) \cos (\kappa_{+2} x) + i ( a_{+3} - a_{+4}) \sin (\kappa_{+2} x) .
			\end{split}
		\end{align}
		Since we search for real solutions, $a_{+1} + a_{+2}$, $ a_{+3} + a_{+4}$ are real and $a_{+1} - a_{+2}$, $a_{+3} - a_{+4}$ are purely imaginary so that $\mathrm{Im}(a_{+1}) = -\mathrm{Im}(a_{+2})$, $\mathrm{Re}(a_{+1}) = \mathrm{Re}(a_{+2})$ and, $\mathrm{Im}(a_{+3}) = -\mathrm{Im}(a_{+4})$, $\mathrm{Re}(a_{+3}) = \mathrm{Re}(a_{+4})$. By the symmetry relations of the imaginary parts of $a_{\pm j}'s$, the matching condition in \eqref{matching condition} gives
		\begin{equation}
			u_+(0) = 2 \mathrm{Re} (a_{+1}) + 2 \mathrm{Re} (a_{+3}) = 0, 
		\end{equation}
		so that $\mathrm{Re}(a_{+1}) = -\mathrm{Re}(a_{+3})$.
		The analogous calculations can be done for $u_{-}(x)$  
		so that 
		\begin{equation}
			\begin{tabular}{p{0.25\linewidth} p{0.25\linewidth}}
				\parbox{\linewidth}{$a_{+1} = -i r_{+1} - s_+ $ \\ $a_{+2} = i r_{+1} - s_+ $ \\ $a_{+3} = -i r_{+2} + s_+ $ \\ $a_{+4} = i r_{+2} + s_+ $}&
				\parbox{\linewidth}{$a_{-1} = -i r_{-1} - s_- $ \\ $a_{-2} = i r_{-2} + s_- $ \\ $a_{-3} = -i r_{-2} + s_- $ \\ $a_{-4} = i r_{-1} - s_- $}
			\end{tabular}
		\end{equation}
		where $r_{\pm j}$ and $s_{\pm}$ are real numbers. Solutions \eqref{eq:u+-sol}  
		are thus of the form
		\begin{align}\label{new u's}
			\begin{split}
				u_+(x) =& 2 r_{+1} \sin(\kappa_{+1}x) + 2 r_{+2} \sin(\kappa_{+2}x)  -2 s_+ \cos(\kappa_{+1}x) + 2 s_+ \cos(\kappa_{+2}x)\\
				u_-(x) =& 2 e^{-\eta x} r_{-1} \sin(\beta x) + 2 e^{\eta x} r_{-2} \sin(\beta x) -2 e^{-\eta x} s_- \cos(\beta x) + 2 e^{\eta x} s_- \cos(\beta x),
			\end{split}
		\end{align} 
		where $i \kappa_{-1} = -\eta + i \beta$ and $i \kappa_{-2} = \eta + i \beta$. Notably, the matching conditions \eqref{matching condition} are satisfied. (For $\nu\geq 1/2$ this form is valid only in the parameter regime $-\nu < \mu < 1-\nu$.) 
		 
		 We are left with the 8 unknowns $L_+$, $\mu$ and $r_{\pm j}$, $s_{\pm}$, $1 \leq j \leq 2$, for the six remaining equations  \eqref{neumann conditions}, \eqref{smoothness condition}. It is convenient to write the latter in the form  
		\begin{equation}\label{system}
			\begin{pmatrix}
				\tilde{u}'_+(L_+) & \begin{matrix} 0 & 0 & 0\end{matrix} \\
				\tilde{u}'''_+(L_+) & \begin{matrix} 0 & 0 & 0\end{matrix} \\
				\begin{matrix} 0 & 0 & 0\end{matrix} & \tilde{u}'_-(-L_-) \\
				\begin{matrix} 0 & 0 & 0\end{matrix} & \tilde{u}'''_-(-L_-) \\
				\tilde{u}'_+(0) & -\tilde{u}'_-(0) \\
				\tilde{u}''_+(0) & -\tilde{u}''_-(0) \\
				\tilde{u}'''_+(0) & -\tilde{u}'''_-(0) 
			\end{pmatrix}
			\begin{pmatrix}
				r_{+1} \\ 
				r_{+2} \\
				s_+ \\
				r_{-1} \\
				r_{-2} \\
				s_-
			\end{pmatrix} = 0
		\end{equation} 
		with $1\times 3$ row vectors $u_+^{(k)}(0)$, $k=1,2,3$, resulting from substituting \eqref{new u's} into \eqref{neumann conditions}, \eqref{smoothness condition} (and removing common constant factors). One readily computes these vectors as 
		\begin{equation}
			\begin{tabular}{p{0.31\linewidth} p{0.5\linewidth}}
				\parbox{\linewidth}{$\tilde{u}'_+(0) = \begin{pmatrix}
						\kappa_{+1}, & \kappa_{+2}, & 0
					\end{pmatrix}$ \\
					$\tilde{u}''_+(0) = \begin{pmatrix}
						0, & 0, & \kappa_{+1}^2 - \kappa_{+2}^2  
					\end{pmatrix}$ \\
					$\tilde{u}'''_+(0) = \begin{pmatrix}
						-\kappa_{+1}^3, & -\kappa_{+2}^3, & 0 
					\end{pmatrix}$
				}&
				\parbox{\linewidth}{$\tilde{u}'_-(0) = \begin{pmatrix}
						\beta, & \beta, & 2\eta
					\end{pmatrix}$ \\
					$\tilde{u}''_-(0) = \begin{pmatrix}
						-2 \eta \beta, & 2 \eta \beta, & 0
					\end{pmatrix}$ \\
					$\tilde{u}'''_-(0) = \begin{pmatrix}
						3 \eta^2 \beta - \beta^3, & 3 \eta^2 \beta - \beta^3, & 2 \eta^3 -6 \eta \beta^2
					\end{pmatrix}$
				}	
			\end{tabular}
		\end{equation}
		as well as  
		\begin{align}
			\begin{split}
				\tilde{u}'_+(L_+) =& \begin{pmatrix}
					\kappa_{+1} \cos(\kappa_{+1} L_+), & \kappa_{+2} \cos(\kappa_{+2} L_+), & \kappa_{+1} \sin(\kappa_{+1} L_+) - \kappa_{+2} \sin(\kappa_{+2} L_+)
				\end{pmatrix} \\
				\tilde{u}'''_+(L_+) =& \begin{pmatrix}
					-\kappa_{+1}^3 \cos(\kappa_{+1} L_+), & -\kappa_{+2}^3 \cos(\kappa_{+2} L_+), & -\kappa_{+1}^3 \sin(\kappa_{+1} L_+) + \kappa_{+2}^3 \sin(\kappa_{+2} L_+)
				\end{pmatrix} \\
			\end{split}
		\end{align}
		and, using $L_- = \pi - L_+$, and abbreviating $c_+:=\cos (\beta(L_+ - \pi)), s_+:=\sin (\beta(L_+ - \pi))$, $e_\pm:=e^{\pm\eta (L_+ - \pi)}$, 
		\begin{align}
			\begin{split}
				\tilde{u}'_-(-L_-) =&\left(\beta e_- c_+-\eta e_- s_+, 
				 \beta e_+ c_++\eta e_+ s_+, 
				 (\eta e_-  +\eta e_+) c_+ + (\beta e_-  -\beta e_+) s_+\right) \\
				\tilde{u}'''_-(-L_-) =& \big(
				(3\eta^2 \beta e_-  - \beta^3 e_-) c_+ - ( \eta^3 e_-  - 3\eta \beta^2 e_-) s_+, 
				(3\eta^2 \beta e_+  -\beta^3 e_+) c_+ + (\eta^3 e_+ - 3\eta \beta^2 e_+) s_+,\\
				&
				(\eta^3 e_-  - 3\eta \beta^2 e_-  + \eta^3 e_+  - 3\eta \beta^2 e_+) c_+ 
				+ (3\eta^2 \beta e_-  - \beta^3 e_-  - 3\eta^2 \beta e_+  + \beta^3 e_+ ) s_+				
				\big).
			\end{split}
		\end{align}
		Solving first two and last three rows in \eqref{system} yields
		\begin{equation} \label{eq:rs}
		\begin{aligned}
			r_{+1} &= 
			\frac{
					\kappa_{+2} (1-\kappa_{+2}^2) r_{+2} \cos(\kappa_{+2} L_+) 
					+ \kappa_{+1} (1 - \kappa_{+1}^2) s_+ \sin(\kappa_{+1} L_+)  
					+ \kappa_{+2} (\kappa_{+2}^2-1) s_+ \sin(\kappa_{+2} L_+) 
			}{
				\kappa_{+1} (-1 + \kappa_{+1}^2) \cos(\kappa_{+1} L_+)
			},\\
			r_{+2} &= \frac{s_+ \sin(\kappa_{+2} L_+)}{\cos(\kappa_{+2} L_+)},\\
			r_{-1} &= -\frac{\kappa_{+1}^2 - \kappa_{+2}^2}{2 \eta \beta} s_+ + r_{-2},\\
			r_{-2} &= -
			\frac{
					4 \eta^2 s_- -\kappa_{+1}^2 s_+ + \kappa_{+2}^2 s_+ 
					+ 2\eta \kappa_{+1} s_+ \tan(\kappa_{+1} L_+) 
					- 2\eta \kappa_{+2} s_+ \tan(\kappa_{+2} L_+)
			}{4 \eta \beta},\\
			s_- &= -
			\frac{
				\begin{aligned}
					\kappa_{+1}(3\eta^2 -\beta^2 +\kappa_{+1}^2) s_+ \tan(\kappa_{+1} L_+) 
					-\kappa_{+2}(3\eta^2 -\beta^2 +\kappa_{+2}^2) s_+ \tan(\kappa_{+2} L_+)
				\end{aligned}				
			}{4\eta (\eta^2 + \beta^2)}.
		\end{aligned}\end{equation}		
The remaining third and fourth rows in \eqref{system} thus read
		\begin{equation}\label{eq:fs}
		\begin{aligned}
			\frac{e^{\eta (-L_+ + \pi)}s_+}{4 \eta \beta} \bigl( \beta \cos(\beta(L_+-\pi)) \gamma_{31} + \sin(\beta(L_+-\pi)) \gamma_{32} \bigl) =0\\
			\frac{e^{\eta (-L_+ + \pi)}s_+}{4 \eta \beta} \bigl( \beta \cos(\beta(L_+-\pi)) \gamma_{41} + \sin(\beta(L_+-\pi)) \gamma_{42} \bigl) =0
		\end{aligned}
		\end{equation}
		where $\eta$, $\beta$ and $e^{\eta(-L_+ + \pi)}$ are non-zero. 
		Hence, \eqref{eq:fs} equivalently reduces to \eqref{reduced system} as claimed. 
		If $s_+ = 0$, then $s_{-}$ and all $r_{\pm j}$ vanish, which implies $u(x)=0$ so that no nontrivial periodic solution exists in this case.
	\end{proof}
		
The approach used in the previous calculation for the region $-\nu < \mu < 1-\nu$ with $\nu>1/2$ can be applied to the remaining parameter regimes for an analogous reduction to algebraic equations. However, it seems difficult to prove the existence and uniqueness of solutions to the algebraic equation system, except for some regimes discussed below. 

Numerical computations suggest that no solutions exist for $\mu > 1-\nu$. 
In the range region $-\nu < \mu < 1-\nu$, $\nu>1/2$, and analogously elsewhere, we have numerically checked the minimum values, for various fixed $\nu$, of the mapping
		\begin{equation}
			H_\nu (\mu,L_+) = f_\nu^2 (\mu,L_+)+ g_\nu^2 (\mu,L_+),
		\end{equation}
with $f_\nu (\mu,L_+)$, $g_\nu(\mu,L_+)$ from \eqref{reduced system}. 
This suggests a unique solution $(\mu^*,L_+^*)$ for each fixed parameter $\nu>1/2$. In particular, for $\nu=1.6$, and within the parameter range $\mu \in (-\nu,1-\nu)$ and $L_+ \in (0,\pi)$, the unique global minimum of $H_\nu$ occurs at $\mu^* \approx -1.1$  
and $L_+^* \approx 2.1$.  
Figure \ref{figure: mu star and L plus star}(a,b) illustrate the corresponding numerical solutions of \eqref{reduced system}; here computed by numerical continuation using \textsc{Auto} \cite{DoedelOldeman2021}. We recall that the value of $\mu^*$ plotted in Figure \ref{figure: fold continuation}(a) has been obtained by continuing fold points of the associated steady state ODE system for $\alpha>1$, which yields the same value.
		
		 \begin{figure}
		 	\centering
		 	\begin{tabular}{ccc}
		 		\includegraphics[width=0.31\textwidth]{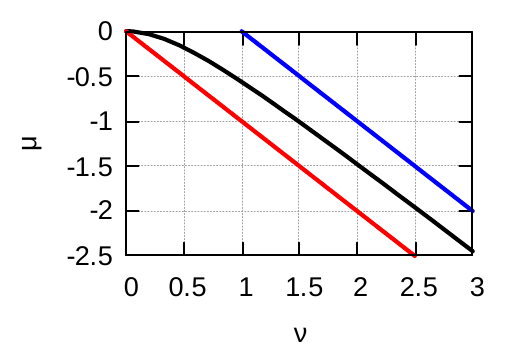} 
				& \includegraphics[width=0.31\textwidth]{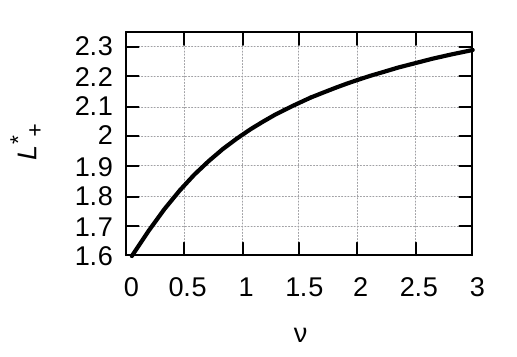} 
				& \includegraphics[width=0.32\textwidth]{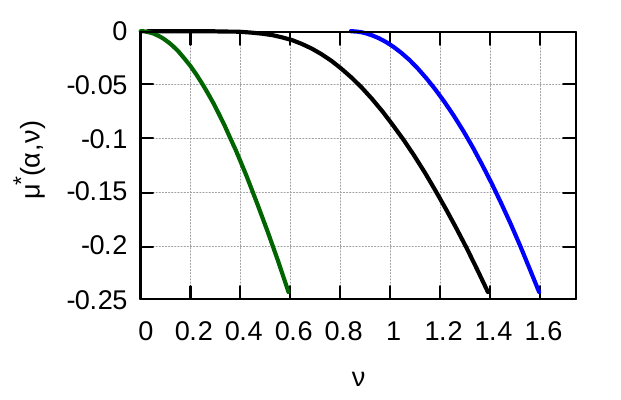}\\
		 		(a) & (b) & (c) 
		 	\end{tabular}
		 	\caption{ (a) Bifurcation point $\mu^*$ as a function of $\nu$ (black, $\alpha=1$), together with the reference lines $\mu=-\nu$ (red) and $\mu=1-\nu$ (blue). (b) Half-period $L_+^*$ as a function of $\nu$. Both quantities are computed from the reduced algebraic system \eqref{reduced system}.(c) Location of fold points $\mu^*(\alpha,\nu)$ as a function of the parameter $\nu$ for different values of the nonlinearity exponent in \eqref{e:spatialNSSHE}: $\alpha = 1.0001$ (green), $\alpha = 1.7$ (black), and $\alpha = 2$ (subcritical, blue). The plot illustrates how the position of the fold point shifts with $\nu$ and $\alpha$ for $2\pi$-periodic branch. }
		 	\label{figure: mu star and L plus star}
		 \end{figure}
	
As discussed in \S\ref{s:rollsalphabigger1} the bifurcation points $\mu^*=\mu^*(\nu)$ is the limit of fold points $\mu^*(\alpha,\nu)$ as $\alpha \searrow 1$. Figure \ref{figure: mu star and L plus star}(c) shows how the location of fold points $\mu^*(\alpha,\nu)$ varies with the parameters $\nu$ for a few values of $\alpha$. For $\alpha = 2$, the bifurcation is, as discussed after \eqref{e:bifeqal2}, subcritical for $\nu>\frac{27}{38}$ so that $\nu=\frac{27}{38}$  is an endpoint of this branch. For $\alpha = 1.7$, the branch connects to $\mu = 0$ along a nearly horizontal line, indicating the absence of a threshold for $\nu$ and uniform criticality across different values of $\nu$. Finally, the case $\alpha = 1.0001$ demonstrates that the transition to $\alpha = 1$ is consistent with the behaviour shown in Figure \ref{figure: mu star and L plus star}(a).

\medskip
The reduction to \eqref{reduced system} is valid also for $\nu=0$, where the system \eqref{SH1} is smooth and linear. In that case the $2\pi$-periodic solutions exist precisely at $\mu=0$. In the following we are interested in tracking these solutions analytically for $\nu\approx 0$, thus highlighting the transition to non-smoothness; there is some analogy to \cite{BettiniCenedeseHaller2024}. This transition occurs at the left boundary in the numerical computations shown in Figure~\ref{figure: mu star and L plus star}(a), which suggest a single branch of solutions exists for $\nu>0$ that has $\mu<0$ and $L_+>\pi/2$. 
		 Next we confirm this rigorously and denote $\ell:=L_+-\pi/2$. Recall that if $u$ is a solution to \eqref{SH1} for some $\mu,\nu$, then $-u$ is a solution for $\mu,-\nu$, i.e., we have the symmetry $(u,\nu,\mu) \to (-u,-\nu,\mu)$. Hence, any branch of solutions gives (typically) another branch upon reflection via this symmetry. Also recall that solutions to \eqref{SH1} come in linear half-spaces parameterised by an amplitude factor. 
		 \begin{theorem}\label{thm:smallnu}
		 	There is $C>0$ such that for all $\nu,\mu,\ell$ with $|\nu|+|\mu|+|\ell|< C$, solutions to \eqref{reduced system} require $-\nu<\mu<\nu$ and form a smooth curve with one endpoint at the origin, and satisfy $\nu>0$, $\mu<0$, $\ell>0$. More specifically, parameterised by $0<\rho_-\ll1$, these solutions have the form 
		 	\begin{align}\label{e:nu0bif}
		 		\begin{split}
		 			\nu &= \rho_-^2 - \frac{30+\pi^2}{48}\rho_-^4 + O(\rho_-^5)\\
		 			\mu &= - \frac{30+\pi^2}{48}\rho_-^4 + O(\rho_-^5),\\
		 			\ell &= \frac{3\pi}{16}\rho_-^2 - \frac{7\pi^2}{64}\rho_-^3 + O(\rho_-^4).
		 	\end{split}\end{align}
		 	Each of these corresponds to a linear half-space of symmetric periodic solutions with two sign changes over one period in the form \eqref{e:bvp1}-\eqref{smoothness condition} with $L_++L_-=\pi$, and there are no other such spaces for $|\nu|+|\mu|+|\ell| < C$.
		 \end{theorem}
		 \begin{proof}
		 	Since Lemma~\ref{lem:globaldecay} also applies to the case without cubic term, there are no non-trivial periodic solutions with $\mu\leq-\nu$. For $-\nu<\mu<\nu$, recall from \eqref{eigenvalues} that $\kappa_{+1} = \sqrt{1 + \sqrt{\nu+\mu}}$ and $\kappa_{+2} = \sqrt{1 - \sqrt{\nu+\mu}}$ are real valued, while $\kappa_{-1} = \beta+ \rmi\eta=\sqrt{1 + \rmi\sqrt{\nu-\mu}}$, $\kappa_{-2} = \beta-\rmi\eta=\sqrt{1 - \rmi\sqrt{\nu-\mu}}$ with $\beta,\eta\in\R$.
		 	To study the bifurcation from $\nu=\mu=\ell=0$, where $\kappa_{+1}=\kappa_{+2}=1$, we regularize $\gamma_{ij}$, and thus \eqref{reduced system}, at $\nu=\mu=\ell=0$ by replacing $\gamma_{ij}$ with 
		 	\[\check{\gamma}_{ij}:=\cos(\kappa_{+1} L_+)\cos(\kappa_{+2}L_+)\gamma_{ij},\] 
		 	$i=3,4, j=1,2$. Next, we introduce $\rho_\pm:=\sqrt{\nu \pm \mu}$, as replacement of the parameters $\nu,\mu$ for the moment, so that $\check{\gamma}_{ij}$ smoothly depends on $\rho_\pm, \ell$ near zero.  In particular, inspecting the formulas for $\gamma_{ij}$ yields the factorisation $\check{\gamma}_{ij} = \rho_+\rho_- \tilde\gamma_{ij}$, where $\tilde\gamma_{ij}$ are smooth in $\rho_\pm$. Hence, for solutions to \eqref{reduced system} with $\nu\geq|\mu|>0$ we can equivalently study 
		 	\begin{align}
		 		\tilde\gamma_{31} \tilde\gamma_{42} - \tilde\gamma_{41} \tilde\gamma_{32} &= 0 \label{e:gammareduced1} \\
		 		\beta \cos\big(\beta (\ell -\pi/2)\big) \tilde\gamma_{31} + \sin\big(\beta (\ell-\pi/2)\big) \tilde\gamma_{32} &= 0.\label{e:gammareduced2}
		 	\end{align}
		 	Taylor expansion gives 
		 	$\eta= \frac 1 2 \rho_- - \frac 1{16} \rho_-^3 + O(\rho_-^5), \beta = 1 + \frac 1 8\rho_-^2  + O(\rho_-^4)$ 
		 	and \eqref{e:gammareduced2} expands (up to a constant factor) as 
		 	\begin{align}
		 		-32  \ell +   8 \pi \ell \rho_-+ \pi \rho_-^2 + 5 \pi  \rho_+^2 &= O(|\ell|^3 + |\rho_+|^3 + |\rho_-|^3),\label{e:gamredexp2} 
		 	\end{align}
		 	which can be solved for $\ell$ by the implicit function theorem as 
		 	\[
		 	\ell = \ell(\rho_-,\rho_+) := \frac{\pi}{32}(\rho_-^2+5 \rho_+^2) + O(|\rho_\pm|^3).
		 	\]
		 	Substituting into \eqref{e:gammareduced1} and expanding gives (up to a constant factor)
		 	
		 	\begin{align}\label{e:gammareduced3}\begin{split}
		 			\rho_-^2-\rho_+^2 
		 			-\frac{\pi}{2}\rho_-^3 
		 			+\frac{\pi}{2}  \rho_- \rho_+^2
		 			+\left(\tfrac{\pi ^2 }{6} -\tfrac{9}{32}\right)\rho_-^4
		 			+\left(\tfrac{\pi ^2}{8}-\tfrac{25}{32}\right)\rho_+^4
		 			-\left(\tfrac{3}{16}+\tfrac{\pi ^2}{3} \right) \rho_-^2 \rho_+^2
		 			= O(|\rho_+|^5 + |\rho_-|^5),
		 	\end{split}\end{align}
		 	whose Newton polygon has a single segment from the leading part $\rho_-^2-\rho_+^2$. The latter has the unique solution $\rho_+=\rho_-$ in the relevant region $\rho_\pm>0$ for solutions to \eqref{reduced system}. This solution can be found in the form $\rho_+=\rho_- \varphi(\rho_-)$ by the implicit function theorem and its expansion reads 
		 	\[
		 	\varphi(\rho_-) = 1 -\left(\frac{5}{8}+\frac{\pi^2}{48}\right)  \rho_-^2 + O(|\rho_-|^3).
		 	\]
		 	It turns out that the linear coefficient of $\varphi$ is the sum of the cubic coefficients in \eqref{e:gammareduced3} and thus vanishes, while the quadratic coefficient is the sum of the quartic coefficients times $1/2$. From this and $\nu = \left(\rho_+^2+ \rho_-^2\right)/2$, $\mu = \left(\rho_+^2- \rho_-^2\right)/2$ we recover the corresponding values $\nu, \mu$ and the relation to $\ell$ as claimed in \eqref{e:nu0bif}.
		 	
		 	\medskip
		 	Finally, we show that there are no solutions in the region $\mu>\nu$ with $\mu\pm\nu<1$. 
		 	In this case the spatial eigenvalues form two purely imaginary complex conjugate pairs and, compared to the previous discussion,  \eqref{eigenvalues} turns into
		 	\begin{equation}\label{eigenvalues2}
		 		\begin{tabular}{p{0.25\linewidth} p{0.25\linewidth}}
		 			\parbox{\linewidth}{$\kappa_{+1} = \sqrt{1 + \sqrt{\mu + \nu}}\in\R$ \\ $\kappa_{-1} = \sqrt{1 +  \sqrt{\mu - \nu}}\in\R$}&
		 			\parbox{\linewidth}{$\kappa_{+2} = \sqrt{1 - \sqrt{\mu + \nu}}\in\R$ \\ $\kappa_{-2} = \sqrt{1 - \sqrt{\mu - \nu}}\in\R$,}	
		 		\end{tabular}
		 	\end{equation} 
		 	and solutions are of the form 
		 	\begin{align*}
		 		u_+ &= s_+ \left(\cos(\kappa_{+1} x)-\cos(\kappa_{+2} x) \right) +  r_{+1}\sin(\kappa_{+1} x) + 
		 		r_{+2}\sin(\kappa_{+1} x)\\
		 		u_- &= s_-\left(\cos(\kappa_{-1} x)-\cos(\kappa_{-2} x)\right) + r_{-1}\sin(\kappa_{-1} x) + r_{-2}\sin(\kappa_{-2} x).
		 	\end{align*}
		 	Analogous to Theorem \ref{theorem: reduced system nu>1/2} we can reduce the boundary value problem for the relevant periodic orbits \eqref{eq:u+-sol}  to algebraic equations, for $s_+> 0$, $L_-=\pi - L_+$,
		 	\begin{align}\label{eq:reduction2}
		 		\begin{split}
		 			\sqrt{\mu-\nu} \left(\kappa_{+1} \tan(L_+ \kappa_{+1}) - \kappa_{+2} \tan(L_+ \kappa_{+2})\right) 
		 			- \sqrt{\mu+\nu} \left(\kappa_{-1} \tan(\kappa_{-1} L_-) - \kappa_{-2} \tan(\kappa_{-2} L_-)\right)&=0\\			
		 			\sqrt{\mu-\nu} \left(\kappa_{+1}^3 \tan(L_+ \kappa_{+1}) - \kappa_{+2}^3 \tan(L_+ \kappa_{+2})\right) 
		 			- \sqrt{\mu+\nu} \left(\kappa_{-1}^3 \tan(\kappa_{-1} L_-) - \kappa_{-2}^3 \tan(\kappa_{-2} L_-)\right)&=0.
		 	\end{split}\end{align}
		 	If $\mu=\nu$ the form of $u_-$ implies $L_-=\pi/2$ to satisfy the boundary conditions so that $L_+=\pi/2$. But then $u_+$ satisfies the boundary conditions for $\mu=-\nu$ only, which means $\nu=\mu=0$. Similarly for $\mu=-\nu$ so that these cases can be excluded, and we can assume $\kappa_{+1}\neq \kappa_{+2}$ and $\kappa_{-1}\neq \kappa_{-2}$.
		 	
		 	Regularizing \eqref{eq:reduction2} at $L_+=\pi/2$ by multiplication with the cosine denominators in the tan-terms, the resulting system possesses the two solutions (i) $\kappa_{+1}L_+=\pi/2$, $\kappa_{-2}L_-=-\pi/2$, and (ii) $\kappa_{+2}L_+=\pi/2$, $\kappa_{-1}L_-=-\pi/2$, which are singularities of \eqref{eq:reduction2}.  By the previous, no other such singularities yield solutions to the regularized equations. We next show that the regularized equations have precisely two solutions, which therefore must be (i), (ii) and thus there are no solutions to \eqref{eq:reduction2} in this region, which concludes the proof.
		 	
		 	We proceed analogous to the case above, and expand \eqref{eq:reduction2} in terms of $\tau_\pm:=\sqrt{\mu\pm\nu}\geq 0$ (replacing $\rho_\pm$) and $\ell=L_+-\pi/2$. After removing common factors $\tau_+\tau_-$, both equations have the same leading order terms given by the sign indefinite $-\frac{32}{\pi^2}\ell^2 + \tau_+^2 + \tau_-^2$. Being quadratic, by the Weierstrass preparation theorem, there are at most two solutions to each equation near zero in terms of, e.g. $\ell$. The scaling ansatz $\tau_\pm = \ell\cdot  t_\pm$ and subsequent division by $\ell^2$ indeed yields these two solutions for the second equation, which satisfy the requirement $\tau_\pm\geq0$ for positive and negative sign of $\ell$ only, respectively. Substituting these solutions into the first equation and ignoring the sign condition on $\ell$ again gives two solutions, where again the requirement $\tau_\pm\geq0$ selects the sign of $\ell$. Hence, when implementing the sign selection from solving the second equation, there are in total exactly two admissible solutions, which must be (i) and (ii) as noted above and thus do not yield solutions to \eqref{eq:reduction2}. 
		 	
		 	\medskip
		 	In conclusion, there is a unique branch of solutions for $\mu,\nu,\ell\approx 0$ as claimed. These solutions to the algebraic equations indeed give nontrivial solutions $u_\pm$ to \eqref{new u's}, and since any such solution must correspond to a solution to \eqref{e:gammareduced1}, \eqref{e:gammareduced2}, these are the only ones.  
		 \end{proof}

\subsection{Roll Solutions of Different Periods for $\alpha =1$} \label{subsec: different periods}
			
In this section we study, for $\alpha=1$, the bifurcation points corresponding to $\mu^*(\nu)$ for rolls of general period. Let us change the periodicity condition to
		\begin{equation}
			L_+ + L_- = l \pi
		\end{equation}
where $l>0$. 

For the analogue of Theorem \ref{theorem: reduced system nu>1/2} we define, for $\nu > 1/2$,  
		\begin{align} \label{eq: gammas_l}
			\begin{split}
				\gamma_{31} (l)=& (-1 + e^{2 \eta (L_+ -l \pi)}) (\kappa_{+1}^2 - \kappa_{+2}^2)-2 \eta (1 + e^{2 \eta (L_+ -l \pi)}) \kappa_{+1} \tan(\kappa_{+1} L_+) \\
				&  +2 \eta (1 + e^{2 \eta (L_+ -l \pi)}) \kappa_{+2} \tan(\kappa_{+2} L_+) \\
				\gamma_{32} (l) =& \eta (1 + e^{2 \eta (L_+ -l \pi)}) (\kappa_{+1}^2 - \kappa_{+2}^2) + (-1 + e^{2 \eta (L_+ -l \pi)}) \kappa_{+1} (\eta^2 - \beta^2 + \kappa_{+1}^2) \tan(\kappa_{+1} L_+) \\
				& - (-1 + e^{2 \eta (L_+ -l \pi)}) \kappa_{+2} (\eta^2 - \beta^2 + \kappa_{+2}^2) \tan(\kappa_{+2} L_+) \\
				\gamma_{41} (l) =& (3 \eta^2 - \beta^2)(-1 + e^{2 \eta (L_+ -l \pi)})(\kappa_{+1}^2 - \kappa_{+2}^2) +2 \eta (1 + e^{2 \eta (L_+ -l \pi)}) \kappa_{+1}^3 \tan(\kappa_{+1} L_+) \\
				& -2 \eta (1 + e^{2 \eta (L_+ -l \pi)}) \kappa_{+2}^3 \tan(\kappa_{+2} L_+) \\
				\gamma_{42} (l)=& \eta (\eta^2 -3\beta^2) (1 + e^{2 \eta (L_+ -l \pi)}) (\kappa_{+1}^2 - \kappa_{+2}^2) \\
				& + (-1 + e^{2 \eta (L_+ -l \pi)}) \kappa_{+1} \big( \eta^4 + \beta^4 -\beta^2 \kappa_{+1}^2 + \eta^2 (2\beta^2 + \kappa_{+1}^2) \big) \tan(\kappa_{+1} L_+) \\
				& - (-1 + e^{2 \eta (L_+ -l \pi)}) \kappa_{+2} \big( \eta^4 + \beta^4 -\beta^2 \kappa_{+2}^2 + \eta^2 (2\beta^2 + \kappa_{+2}^2) \big) \tan(\kappa_{+2} L_+)
			\end{split}
		\end{align}
		where $\eta = \mathrm{Re}(i \kappa_{-2}) = -\mathrm{Re} (i \kappa_{-1}) $, $\beta = \mathrm{Im}(i \kappa_{-2}) = \mathrm{Im} (i \kappa_{-1})$
and proceed as in the proof of Theorem \ref{theorem: reduced system nu>1/2}. Omitting details, the reduced system analogous to \eqref{reduced system} becomes 
		\begin{align}\label{e:reduced_l}
			\begin{split}
				& f_\nu (l; \mu,L_+) =  \gamma_{31}(l) \gamma_{42}(l) - \gamma_{41}(l) \gamma_{32}(l) = 0 \\
				& g_\nu(l; \mu,L_+) =  \beta \cos\big(\beta (L_+ -l \pi)\big) \gamma_{31} (l) + \sin\big(\beta (L_+ - l \pi)\big) \gamma_{32} (l) = 0.
			\end{split}
		\end{align} 
We plot numerical solutions of this system in Figure \ref{figure: different periodic rolls}. 
		\begin{figure}
			\centering
			\begin{tabular}{cc}
				\includegraphics[width=0.44\linewidth]{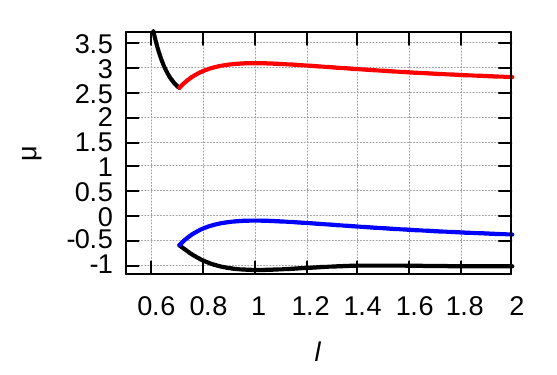} &
				\includegraphics[width=0.44\linewidth]{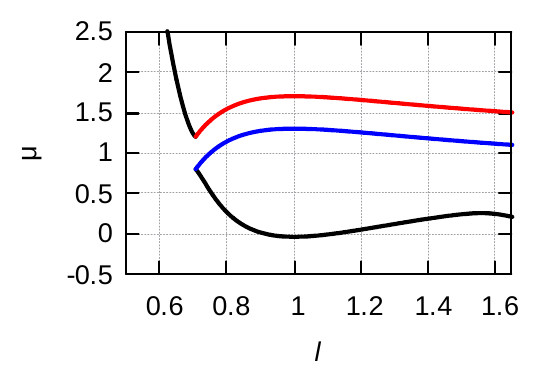} \\
				(a) & (b)
			\end{tabular}
			\caption{ $\mu$-coordinate of the solutions of the reduced algebraic system \eqref{e:reduced_l} as a function of $l$, corresponding to the bifurcation points of $2l\pi$-periodic roll solutions for (a) $\nu=1.6$ and (b) $\nu=0.2$. The black branches correspond to sign-changing solutions within $2l\pi$-period. The blue branch corresponds to the positive-sign solution of \eqref{eq:u_+0 char eq} with $k=1/l$, while the red branch corresponds to the negative-sign solution of \eqref{eq:u_+0 char eq} with $\tilde{\mu}=\mu-\nu$. In both panels, the black branches exhibit a discontinuity at $l=1/\sqrt{2}$: the left branch terminates at $\mu=1+\nu$, while the right branch emerges at $\mu=1-\nu$, cf. Theorem~\ref{thm:lowl}. 
			}
			\label{figure: different periodic rolls}
		\end{figure}		

Let us denote the associated $\mu$-value as $\mu^*_l(\nu)$. The global minimum of $l\mapsto\mu^*_l(1.6)$ corresponds to wavelength $2\pi$ so that, for increasing $\mu$ these bifurcate first; we find this also for other values of $\nu$. This special role of wavelength $2\pi$ is clear in the differentiable case $\alpha>1$ at $\mu=0$ due to the spatial eigenvalues of the linearization. But for $\alpha=1$ there are no (unique) spatial eigenvalues, and it would be interesting to find an analytical explanation. 
More specifically, in case $\alpha>1$ the characteristic equation \eqref{e:spatialNSSHEchar} in $u=0$ gives $\mu=(1-k^2)^2$ and thus $k=1/ l$ yields the analogue of $l\mapsto \mu^*_l(\nu)$. Indeed, the graph of the latter shares qualitative features with that in Figure \ref{figure: different periodic rolls}: a global minimum at $l=1$ and the approach of a constant value as $l$ is increasing; in Figure \ref{figure: different periodic rolls} (a) this is $\mu\approx -1.1$ for which it would likewise be interesting to provide an analytical explanation. 
However, differences are that the behaviour for increasing $l$ is monotone and for $l$ below $1$ it features a singularity at $l=0$, while the black curves in Figure \ref{figure: different periodic rolls} each exhibit a discontinuity at $l\approx 1/\sqrt{2}$, and that for larger $l$ is non-monotone. This suggests that for $\alpha=1$ the possible periods of the primary bifurcating branch of rolls from  $u=0$ are bounded from below by $\sqrt{2}\pi$.
In contrast, in the smooth situation for $\alpha>1$ no such discontinuity occurs, and the red and blue branches lie on top of each other. The following theorem provides a partial explanation. 
 
		\begin{theorem}\label{thm:lowl}
			Let $\nu > 1/2$, $\mu = 1-\nu$ and consider \eqref{e:reduced_l} with coefficients given in \eqref{eq: gammas_l}, where the complex frequencies $\kappa_{\pm j}$ are defined in \eqref{eigenvalues}, and $L_+ \in [0,l\pi]$ for $l>0$. Then \eqref{e:reduced_l} has a unique solution $(l, L_+)$ that gives a $2\pi l$-periodic roll to \eqref{SH1} with \eqref{e:bvp1}-\eqref{smoothness condition} with one zero in the interval $[-L_-,L_+]$. This solution is given by
			\begin{equation}
				l = \frac{1}{\sqrt{2}}, \quad L_+ = \pi/\sqrt{2}.
			\end{equation}
			In particular, it is independent of $\nu$, has $L_- = 0$ and the corresponding roll $u_*$ has period $\sqrt{2}\pi$ and satisfies $u_*(x)>0$ for $x\in(0,\sqrt{2}\pi)$ and $u_*(0)=u_*(\sqrt{2}\pi)=0$.
		 \end{theorem}

We emphasize that this result proves that one segment of the branch of sign-changing roll solutions for any $\nu$ contains $(l,L_+)=\left(1,\frac{\pi}{\sqrt{2}}\right)$, but it does not provide information about the nature of the bifurcation at this point or about the other segment. However, by the same arguments, one can show that the other segment contains $(l,L_-)=\left(1,\frac{\pi}{\sqrt{2}}\right)$, with $L_+=0$. 
The fact that $u_*(x)>0$ for $x\in(0,\sqrt{2}\pi)$ is natural for a termination point of  bifurcations from $u=0$, as it suggests the branch transitions to rolls that accompany $u_{+0}$ in a reversible Hopf-zero bifurcation, cf.\ \S\ref{sec:homostatesevals}. Indeed, the spatial frequency $k_+$ from \eqref{e:kpm} evaluated in $u=u_{+0}$ converges to $\sqrt{2}$ as $\alpha\searrow 1$, which gives the limiting period $\sqrt{2}\pi$ for accompanying rolls in $\{u>0\}$.

		\begin{proof}			
			Since $\nu > 1/2$ and $\mu = 1-\nu$, the complex frequencies $\kappa_{\pm j}$ given in \eqref{eigenvalues} yield
			\begin{equation}\label{eq: eigenvalues_1-nu}
				\kappa_{+1}= \sqrt{2}, \quad \kappa_{+2} = 0, \quad \kappa_{-2} = \sqrt{1 - \sqrt{1-2\nu}},
			\end{equation}
			where $i\kappa_{-2} = \eta + i\beta$. For notational convenience, we set 
			\begin{equation}
				X = -1 + \exp\bigl(2\eta(L_+ - l\pi)\bigr)
				\text{ and }
				Y = 1 + \exp\bigl(2\eta(L_+ - l\pi)\bigr).
			\end{equation}
			The coefficients $\gamma_{ij}(l)$ in \eqref{eq: gammas_l} may then be expressed in terms of $X$ and $Y$. Substituting \eqref{eq: eigenvalues_1-nu} and collecting like terms yields
			\begin{align} \label{eq: f_1-nu}
				\begin{split}
					f_\nu (l; \mu,L_+) =\;& XY(-8\beta^2\eta - 8\eta^3) \\
					&+ X^2(-8\sqrt{2}\eta^2 + 12\sqrt{2}\beta^2\eta^2 - 4\sqrt{2}\eta^4)\tan(\sqrt{2}L_+) \\
					&+ Y^2(-8\sqrt{2}\eta^2 + 12\sqrt{2}\beta^2\eta^2 - 4\sqrt{2}\eta^4)\tan(\sqrt{2}L_+) \\
					&+ XY(-16\eta + 16\beta^2\eta - 4\beta^4\eta - 16\eta^3 - 8\beta^2\eta^3 - 4\eta^5)
					\tan^2(\sqrt{2}L_+).
				\end{split}
			\end{align}
			Similarly, we find
			\begin{align}\label{eq: g_1-nu}
				\begin{split}
					g_\nu (l; \mu,L_+) = &Y\Bigl(2\eta \sin\bigl((L_+ - l\pi)\beta\bigr)
					- 2\sqrt{2}\,\beta\eta \cos\bigl((L_+ - l\pi)\beta\bigr)\tan(\sqrt{2}L_+)\Bigr) \\
					&+ X\Bigl(2\beta \cos\bigl((L_+ - l\pi)\beta\bigr)
					+ 2\sqrt{2}\sin\bigl((L_+ - l\pi)\beta\bigr)\tan(\sqrt{2}L_+) \\
					&\qquad - \sqrt{2}\beta^2 \sin\bigl((L_+ - l\pi)\beta\bigr)\tan(\sqrt{2}L_+)
					+ \sqrt{2}\eta^2 \sin\bigl((L_+ - l\pi)\beta\bigr)\tan(\sqrt{2}L_+)\Bigr).
				\end{split}
			\end{align}
			
			 Let us consider a fixed $l>0$ and seek possible solutions $L_+\in [0,l\pi]$. We first consider the boundary cases $L_+=0, l\pi$. For $L_+ = 0$ we have $X,Y\neq 0$ and \eqref{eq: f_1-nu} yields
			\begin{equation*}
				f_\nu (l; \mu,0)= -2 X Y \eta (\beta^2 + \eta ^2),
			\end{equation*}  
			which vanishes if and only if $\beta^2 + \eta ^2 = 0$ or $\eta =0$. We have $\beta^2 + \eta ^2 = (i \kappa_{-2}) (-i \kappa_{-2}) = \kappa_{-2}^2 = 1-\sqrt{1-2\nu} =0$ if and only if $\nu =0$, which contradicts the assumption $\nu > 1/2$. Independent of $L_+$, that $\nu > 1/2$ implies $\kappa_{-2}^2 =1-\sqrt{1-2\nu}$ is non-real, so that from $\kappa_{-2}^2 = (\beta - i \eta)^2 = \beta^2 - \eta^2 - 2i \beta \eta$ it follows that $\eta \neq 0$ and $\beta \neq 0$. Therefore, for $L_+ =0$, no solution exists.
			
			Next, we set $L_+ = l\pi$ so that $L_-=0$.  Substituting this into \eqref{eq: f_1-nu} and \eqref{eq: g_1-nu}, we obtain
			\begin{align*}
				f_\nu (l; \mu,L_+) 
				&= -16\sqrt{2}\eta^2 (2 - 3\beta^2 + \eta^2)\tan(\sqrt{2}l\pi),\\
				g_\nu (l; \mu,L_+) 
				&= -4\sqrt{2}\beta\eta\tan(\sqrt{2}l\pi) .
			\end{align*}
			Since $\beta\eta\neq 0$ as shown above, $f_\nu(l;\mu,L_+)$ and $g_\nu(l;\mu,L_+)$ vanish if and only if $\tan(\sqrt{2}l\pi)=0$. 
			Since $l>0$, the smallest positive admissible value is $l=1/\sqrt{2}$. Checking the coefficients in \eqref{eq:rs},  
			we obtain a nontrivial solution to the boundary value problem. The larger solutions $l=j/\sqrt{2}$, $j\in\mathbb{N}$ correspond to extending the roll solution $u$, which thus has additional zeros in $[0,l\pi]$.
			
			Now, we consider $L_+\in(0,l\pi)$ and set $L_+ = m l \pi$ for some $m\in(0,1)$ so that $L_- \neq 0$. In this case, \eqref{eq: f_1-nu} and \eqref{eq: g_1-nu} reduce to  
			\begin{equation*} 
				\begin{aligned}
					f_\nu (l; \mu,m l\pi) &=-16\sqrt{2}\eta^2  \left(2 -3\beta^2 + \eta^2\right) \tan\big(\sqrt{2} l m \pi\big)
					 \\
					g_\nu (l; \mu,m l\pi) &= 4\eta \sin\big(\beta l(-\pi + m\pi)\big)
					-4\sqrt{2}\beta\eta \cos\big(\beta l(-\pi + m\pi)\big)\tan\big(\sqrt{2}lm\pi\big).
				\end{aligned}
			\end{equation*}
			which  
			simultaneously vanish if and only if $\tan\big(\sqrt{2} l m \pi\big) =0$ and $\sin\big((-l\pi + lm\pi)\beta\big)=0$; again recall $\eta\beta\neq 0$.  
			Consequently, the admissible values of $l>0$ must satisfy either
			\begin{equation} \label{eq: case_1}
				l = \frac{c_1}{\sqrt{2}m} \text{ and } l = \frac{2 c_2}{\beta(1-m)} \text{ for } c_1,c_2 \in \mathbb{N}
			\end{equation}
			or 
			\begin{equation}\label{eq: case_2}
				l = \frac{c_1}{\sqrt{2}m} \text{ and } l = \frac{1+2c_2}{\beta(1-m)} \text{ for } c_1,c_2 \in \mathbb{N}.
			\end{equation}
			
			In case \eqref{eq: case_1}, we obtain $m=m (c_1,c_2,\beta) := \frac{c_1 \beta}{c_1 \beta +2 \sqrt{2} c_2}$ for any $c_1 > 0$ and $c_2 > 0$.
			For this value of $m$ we have $l = \frac{c_1}{\sqrt{2}} + \frac{2c_2}{\beta}$, which attains its minimum for $c_1 = c_2 = 1$. However, for these parameter values,  the corresponding solution $u$ attains  additional zeros in $[0,l\pi]$ and hence does not yield an admissible solution. 		
				
			In case \eqref{eq: case_2}, we find that $m=m (c_1,c_2,\beta) := \frac{c_1 \beta}{c_1 \beta +\sqrt{2} +2 \sqrt{2}c_2 }$ for any $c_1 > 0$. 
			We have 
			$l = \frac{c_1}{\sqrt{2}} + \frac{1}{\beta} + \frac{2 c_2}{\beta}$, is minimized when $c_1 =1$ and $c_2 =0$. Nevertheless,  the corresponding solution $u$ again attains  additional zeros in $[0,l\pi]$.  Consequently, $L_+=ml\pi$ does not give rise to an admissible solution in either case.

			Therefore, for $\nu>1/2$ and $\mu=1-\nu$, we conclude that $l =1/\sqrt{2}$  at $L_+=l\pi$. 
		\end{proof}

An analogous argument applied to \eqref{eq:reduction2} shows that the same conclusion holds for $\nu<1/2$;  cf. Figure~\ref{figure: different periodic rolls}(b).

Based on the numerical results we conjecture that for any fixed $\nu$ and any period, the values of $\mu$ for bifurcation points of rolls from $u=0$ are constrained to $(-\nu, 1-\nu)\cup(1+\nu,\infty)$.

\section{Homoclinic Snaking}\label{sec:snaking}
	
Beyond periodic roll solutions, in this section we numerically investigate how the non-smoothness affects spatially localised stationary solutions $u$ to \eqref{nonsmooth SHE}, i.e., solutions that spatially asymptote to zero. These correspond to  
homoclinic trajectories in the spatial ODE phase space $\R^4$ of \eqref{e:spatialNSSHE}. 
 Homoclinic snaking \cite{BurkeKnobloch2006,BurkeKnobloch2007} refers to the oscillatory behaviour of a branch of homoclinic solutions that approaches a heteroclinic cycle containing a periodic orbit corresponding to a roll solution (and its translates). 
During oscillations of the branch, the corresponding localised solutions $u$ successively gain additional rolls in their spatial profile. 
For $\alpha=2$, in order for such localised solutions to be stable in \eqref{nonsmooth SHE}, the stable zero state should co-exist with rolls at $\mu<0$, which requires a subcritical bifurcation of rolls. 
	
		\begin{figure}
			\centering
			\begin{tabular}{cc}
				\includegraphics[width=0.4\linewidth]{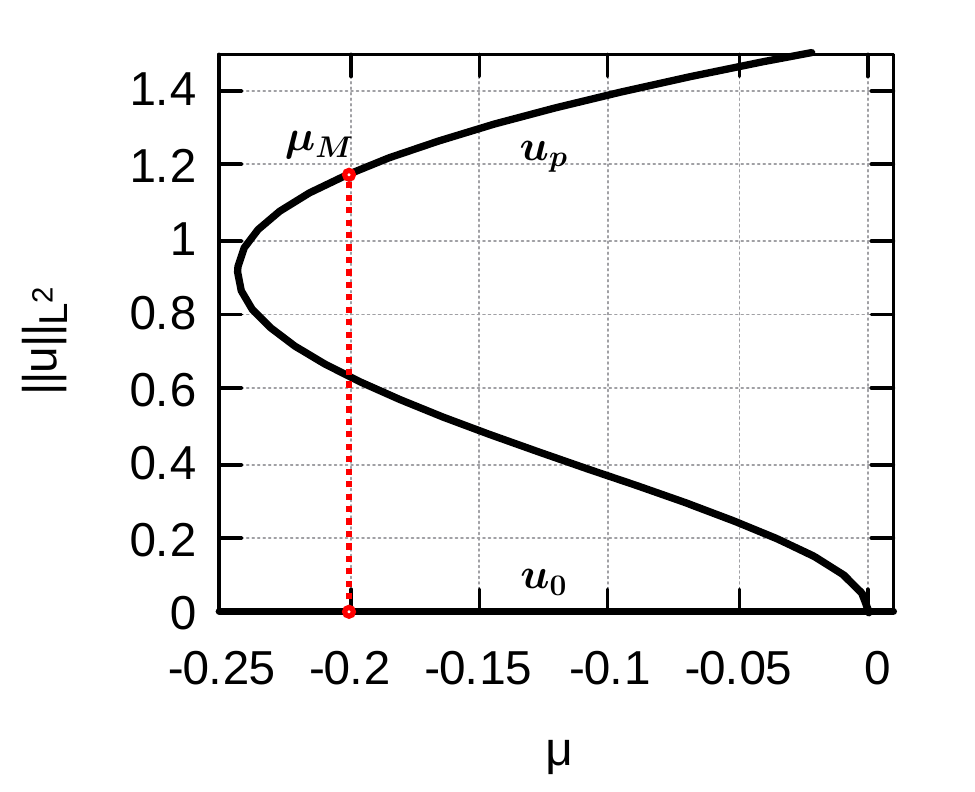} &
				\includegraphics[width=0.44\linewidth]{images/snakingmaxu.pdf} \\
				(a) & (b)
			\end{tabular}
			\caption{ For $\nu = 1.6$: (a) Location of Maxwell point in \eqref{eq:energy_func} for $\alpha=2$. (b) Snaking branches of \eqref{e:spatialNSSHE} for $\alpha = 2$ (magenta), $1.5$ (orange), $1.1$ (green), $1.05$ (blue), $1.04126$ (red), $1.03242$ (black) and the scaled $L^2$-norm of the rolls at the Maxwell points $\mu=\mu_M(\alpha,1.6)$ (purple) for $\alpha \in [1,2]$.}          
			\label{figure: maxwell and snakings}
		\end{figure} 
		
\subsection{Non-smooth birth of homoclinic snaking}\label{subsec: homoclinic snaking}
The variational structure of the Swift-Hohenberg equation \eqref{nonsmooth SHE} yields a Hamiltonian for \eqref{e:spatialNSSHE} given by the energy functional
\begin{equation}\label{eq:energy_func}
E[u] = \int _{\Omega} \frac{1}{2} \bigl( (1+\partial_x^2) u \bigl)^2 - \frac{1}{2} \mu u^2 - \frac{\nu u|u|^\alpha}{\alpha+1} + \frac{1}{4}u^4 \,dx\ .
\end{equation}
Hence, a heteroclinic orbit connecting $u=u_h=0$ and a roll $u_p$ in \eqref{e:spatialNSSHE} requires equal energy, $E[u_p] = E[u_h]$. The associated  parameter value $\mu=\mu_M(\alpha,\nu)$ is referred to as the Maxwell point. The heteroclinic bifurcation structure generically implies that such a heteroclinic connection persists for nearby values of $\mu\approx \mu_M(\alpha,\nu)$, which can be viewed as the result of an energy barrier. This region of $\mu$-values around $\mu_M(\alpha,\nu)$ is generically bounded by fold points and yields the so-called pinning region for the bifurcating homoclinic orbits and their oscillating branch. We refer to \cite{BurkeKnobloch2006,BurkeKnobloch2007,Beck2009} for more details. 
Since the trivial solution satisfies $E[u_h = 0] =0$, we thus search for $\mu$ at which the roll solutions have zero energy. For $\alpha = 2$, $\nu=1.6$ we numerically find that $\mu_M(2,1.6) \approx -0.2$. 
See Figure~\ref{figure: maxwell and snakings}(a). 
		
Concerning \eqref{nonsmooth SHE}, we are interested in the changes to snaking as $\alpha$ decreases to $1$. In Figure~\ref{figure: maxwell and snakings}(b) we plot the resulting modified snaking branches for several values of such $\alpha$. Notably, the width of the snaking region and the norm decrease as $\alpha$ decreases. Indeed, we know from \S\ref{sec:rolls}, cf.\ Figure~\ref{figure: periodic for different alphas}, that the bifurcation of periodic rolls at $\alpha=1$ is supercritical, so that one would expect the snaking disappears at $\alpha=1$. 
This is corroborated by the curve resulting from the Maxwell points plotted in Figure~\ref{figure: all at once}(a) in purple, which connects to zero amplitude at $\mu_M(1,1.6) \approx -1.1$,  
coinciding with the bifurcation point $\mu^*$ from \S\ref{subsec: Bifurcation of periodic rolls}, at which the onset of the $2\pi$-periodic solution occurs. 
Here we note that localised solutions are computed on a bounded interval of length $L = 400$ and the vertical axis is the $L^2$-norm normalised by this interval length. Therefore, the Maxwell point curve lies above the snaking curve.

As mentioned, the disappearance of homoclinic snaking at $\alpha = 1$ can be viewed as a consequence of the switch from subcritical for $\alpha\in(1,2)$ to supercritical roll bifurcation for $\alpha=1$. Such a change in criticality also occurs for fixed $\alpha\in(1,2)$ when $\nu$ changes sign from positive, and for $\alpha=2$ when $\nu$ decreases below $\sqrt{\tfrac{27}{38}}$. However, in the latter case, the location of the bifurcation point is fixed at $\mu=0$ while in the transition from $\alpha>1$ to $\alpha=1$ this point discontinuously jumps from $\mu=0$ to $\mu=\mu^*\in (-\nu,0)$ as discussed in \S\ref{subsec: Bifurcation of periodic rolls}. 

In terms of the branch geometry, we have a transition from a supercritical half-pitchfork of rolls for $\alpha=1$  
to two branches for $\alpha>1$: one folded and one straight. (While this is superficially reminiscent of an imperfect pitchfork, the latter has two branches before and after.)  
The transition from $\alpha=1$ to $\alpha>1$ is a form of singular bifurcation that would be interesting to associate with a classification. It leads to the simultaneous bifurcation of rolls from $u=0$ for all values of $\mu$ in the interval $(\mu^*,0)$. In this sense, the appearance of snaking when increasing $\alpha$ from $\alpha=1$ may be interpreted as a non-smooth birth of homoclinic snaking that is not associated with a local switch from super- to subcriticality. 

\subsection{Homoclinic snaking at $\alpha=1$}
\label{subsec: homoclinic snaking at 1}

 The previous discussion highlights the relevance of the criticality of the roll bifurcation at $\alpha=1$ for homoclinic snaking. In this short section we illustrate how a modification to \eqref{nonsmooth SHE} can maintain subcriticality at $\alpha=1$ and thus prevent the disappearance of snaking. 

Recall the homogeneous steady state equation \eqref{e:homeqal1a}, which yielded a supercritical half-pitchfork at $1-\nu$. For this equation one readily sees that criticality can be changed to a subcritical half-transcritical bifurcation by adding a quadratic term $\rho u^2$ with any $\rho>0$ to the right hand side; for larger $|u|$ the branch then folds according to the dominant cubic term. This motivates adding such an additional quadratic term to \eqref{nonsmooth SHE} and we thus consider 
		\begin{equation}\label{additional rho}
			\partial_t u= -(1+\partial_x^2)^2 u + \mu u + \nu |u|^{\alpha} + \rho u^2 -u^3 ,
		\end{equation}
		with additional parameter $\rho$. The energy associated with the modified equation reads
		\begin{equation*}
			E[u] = \int _{\Omega} \frac{1}{2} \bigl( (1+\partial_x^2) u \bigl)^2 - \frac{1}{2} \mu u^2 - \frac{\nu u|u|^\alpha}{\alpha+1} - \frac{\rho u^3}{3}    + \frac{1}{4}u^4 \,dx,
		\end{equation*}
which admits computing the Maxwell points.  Indeed, for $\rho>0$ we numerically find that the quadratic term also yields subcriticality for rolls, and thus maintains homoclinic snaking at $\alpha=1$. See Figure~\ref{figure: all at once}(b). Notably, there is no threshold on $\rho>0$ to maintain subcriticality just as for the homogeneous steady states mentioned above. We also note that the bifurcation point $\mu^*$ for $\alpha=1$ does not depend on $\rho$. 

		\begin{figure}
			\centering
			\begin{tabular}{cc}
				\includegraphics[width=0.45\textwidth]{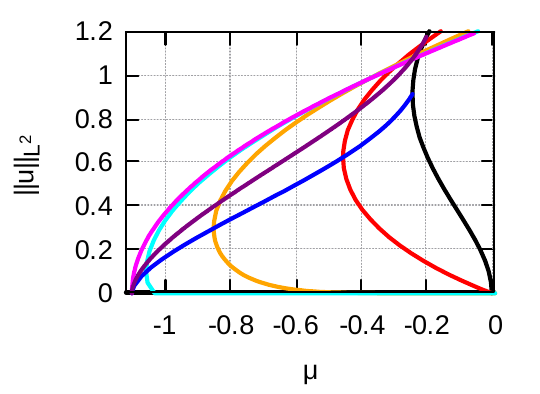} & \includegraphics[width=0.44\textwidth]{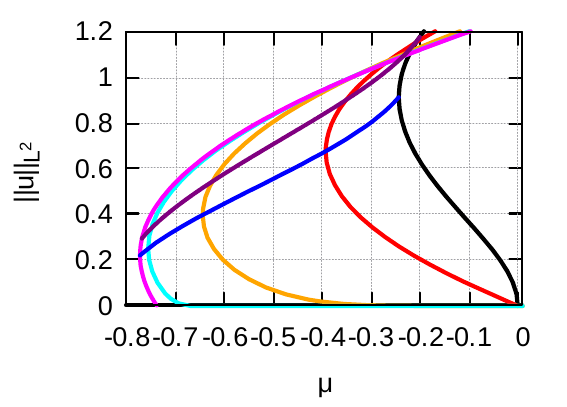} \\
				(a) & (b)
			\end{tabular}
			\caption{$2\pi$-periodic rolls for $\alpha = 2 \text{ (black)}, 1.5 \text{ (red)}, 1.1 \text{ (orange)}, 1.01 \text{ (cyan)}, 1 \text{ (magenta)}$ and the positions of fold points (blue), and Maxwell points (purple) of SHE with nonlinearity (a) $\nu |u|^\alpha - u^3$ with $\nu =1.6$, and (b) $\nu |u|^\alpha +\rho u^2- u^3$ with $\nu = 1.2$ and $\rho = 0.4$, showing that subcriticality persists at $\alpha=1$.}
			\label{figure: all at once}
		\end{figure}

\section{Discussion}\label{sec:discuss}

	We have investigated the one-dimensional generalized quadratic-cubic Swift-Hohenberg equation \eqref{nonsmooth SHE} with quadratic term replaced by $\nu |u|^\alpha$  for $\alpha \in [1,2]$, except in \S\ref{subsec: homoclinic snaking at 1}, where both even terms are present. Our analysis focused on three fundamental classes of steady state solutions: spatially homogeneous states, roll solutions, and homoclinic orbits that interact with rolls in snaking structures. The results reveal how the transition from the classical smooth case $(\alpha=2)$ to the non-differentiable case $(\alpha=1)$ affects the bifurcation structure and the existence of localized states.
	
	Considering the spatially homogeneous solutions in \S\ref{sec:homrolls} has been an instructive, yet straightforward problem. Here the classical transcritical bifurcation for $\alpha=2$ at $\mu=1$ persists for all $\alpha \in (1,2]$. As $\alpha$ decreases from $2$ towards $1$, the nontrivial branch becomes progressively flatter. In the limiting case $\alpha=1$, the transcritical bifurcation is replaced by two supercritical ``half-pitchfork'' bifurcations occurring at $\mu = 1 \pm \nu$. In all cases, the left part of the branch folds back to $\mu\geq 0$. In this transition from $\alpha>1$ to $\alpha=1$, the bifurcation point thus splits into two and jumps, and part of the non-trivial branch exhibits a discontinuous breaking; a signature of the non-differentiable nature of the nonlinearity at $\alpha=1$.
	
	We have found in \S\ref{sec:rolls} that the bifurcation of roll solutions undergoes a similar, but more subtle change. For the classical case $\alpha=2$, rolls with critical wavenumber $k=1$ bifurcate from $\mu=0$ with criticality depending on the sign of $\nu^2-\frac{27}{38}$. In contrast, the bifurcation is subcritical for all $\alpha \in (1,2)$, but still at $\mu=0$, and the non-trivial branch switches from being locally concave to convex for $\nu>0$ below $3/2$. This part becomes increasingly flat as $\alpha$ approaches $1$, and in all cases folds back to $\mu\geq 0$. 
	At $\alpha=1$, the bifurcation point jumps from $\mu=0$ to a value  $\mu^*(\nu)\in(-\nu,0)$, and the flattened part of the non-trivial branch merges with the trivial one, thus yielding a supercritical ``half-pitchfork'' bifurcation of rolls.  The non-smooth setting does not readily permit a standard analytical treatment, in particular at $\alpha=1$, which precludes an explicit formula of $\mu^*$. However, we obtain an expansion for $0<\nu \ll1$ and generally predict the local geometry for $\alpha\in(1,2]$ using suitable estimates, and reduce the problem to a system of algebraic equations. In addition, energy estimates show that the trivial state is stable for $\mu<-\nu$ and unstable for $\mu>0$ so that $\mu^*(\nu)>0$. Moreover, numerical results consistently indicate that the primary bifurcating branch, when increasing $\mu$ from $-\nu$, corresponds to $2\pi$-periodic roll solutions.
	
	The differences in roll bifurcation criticality (subcritical for $\alpha\in(1,2)$, supercritical for $\alpha=1$) directly impact the occurrence of homoclinic snaking in \S\ref{sec:snaking}. In the classical case $\alpha=2$, snaking occurs only when the roll branch is subcritical, i.e., $\nu^2>\frac{27}{38}$, giving rise to the well-known hysteresis mechanism between uniform and patterned states. Consistent with this, homoclinic snaking persists for every $\alpha \in (1,2)$, where the roll branch bifurcation is subcritical, and snaking disappears at $\alpha=1$, where the non-smooth transition to supercritical bifurcation occurs. Further corroborating this geometric relation of subcriticality and snaking, we find that snaking persists at $\alpha=1$ when enforcing a subcritical roll bifurcation via the additional term $\rho u^2$, cf.\ \S\ref{subsec: homoclinic snaking at 1}.

\medskip	
	Several directions are thus motivated for future investigation. 
	Directly connecting to the results, it would be interesting to extend the rigorous analysis and prove $\mu^*(\nu)<0$ (as suggested by the destabilisation from Remark~\ref{rem:energynondecay}),  
	the uniqueness of this bifurcation point for any $\nu$, and the supercriticality. 
	 It would be particularly interesting to understand why for $\alpha=1$, despite the lack of differentiability, the Fourier-symbol still correctly predicts the critical wavenumber $1$ at the primary bifurcation. 
		
	A natural extension of the nonlinearity would be to consider the form 	$\nu |u|^\alpha - u|u|^\beta$, $\alpha \in [1,2]$, $\beta \in [0,2]$, which allows for a broader study of how non-smoothness influences bifurcation structure, pattern formation, and homoclinic snaking through the competition of $\alpha$ and $\beta$. 
	
	Another relevant direction would be the stability analysis of the solutions. In the differentiable range $\alpha\in(1,2)$ this can be done via the dispersion relation, but again at $\alpha=1$ this is much more challenging. In the context of stability, it would be relevant to determine in what way the additional ``ladders''-structure from classical homoclinic snaking behaves. Moreover, how do invasion waves of patterns behave?

	Last but not least, it would interesting to study pattern forming bifurcations in higher space dimensions, in particular squares and hexagons, and their transitions for $\alpha\in[1,2)$.

\section*{Acknowledgement} The authors acknowledge the support by the Deutsche Forschungsgemeinschaft (DFG) within the Research Training Group GRK 2583 "Modeling, Simulation and Optimization of Fluid Dynamic Applications", as well as from funds from the University of Hamburg.

\bibliographystyle{alpha}
\bibliography{sample}

\newcommand{\etalchar}[1]{$^{#1}$}
\begin{thebibliography}{SMCGK23}

\bibitem[AAR99]{andrews1999special}
G.~E. Andrews, R.~Askey, and R.~Roy.
\newblock {\em Special Functions}, volume~71 of {\em Encyclopedia of
  Mathematics and Its Applications}.
\newblock Cambridge University Press, Cambridge, 1999.

\bibitem[AP93]{AmbrosettiProdi1993}
A.~Ambrosetti and G.~Prodi.
\newblock {\em A Primer of Nonlinear Analysis}, volume~34 of {\em Cambridge
  Studies in Advanced Mathematics}.
\newblock Cambridge University Press, Cambridge, 1993.

\bibitem[AW96]{AulbachWanner1996}
B.~Aulbach and T.~Wanner.
\newblock Integral manifolds for {C}arath\'eodory type differential equations
  in {B}anach spaces.
\newblock In {\em Six lectures on dynamical systems ({A}ugsburg, 1994)}, pages
  45--119. World Sci. Publ., River Edge, NJ, 1996.

\bibitem[Bak03]{Bakanas2003}
R.~Bakanas.
\newblock Travelling fronts in a piecewise-linear bistable system.
\newblock {\em Nonlinearity}, 16(1):313--325, 2003.

\bibitem[BK06]{BurkeKnobloch2006}
J.~Burke and E.~Knobloch.
\newblock Localized states in the generalized {S}wift-{H}ohenberg equation.
\newblock {\em Phys. Rev. E (3)}, 73(5):056211, 15, 2006.

\bibitem[BK07]{BurkeKnobloch2007}
J.~Burke and E.~Knobloch.
\newblock Snakes and ladders: localized states in the {S}wift-{H}ohenberg
  equation.
\newblock {\em Phys. Lett. A}, 360(6):681--688, 2007.

\bibitem[BKL{\etalchar{+}}09]{Beck2009}
M.~Beck, J.~Knobloch, D.~J.~B. Lloyd, B.~Sandstede, and T.~Wagenknecht.
\newblock Snakes, ladders, and isolas of localized patterns.
\newblock {\em SIAM J. Math. Anal.}, 41(3):936--972, 2009.

\bibitem[BKPS23]{BelykhKuskePorfiriSimson2023}
I.~Belykh, R.~Kuske, M.~Porfiri, and D.~J.~W. Simpson.
\newblock Beyond the bristol book: Advances and perspectives in non-smooth
  dynamics and applications.
\newblock {\em Chaos: An Interdisciplinary Journal of Nonlinear Science},
  33(1):010402, 01 2023.

\bibitem[CG09]{CrossGreenside2009}
M.~Cross and H.~Greenside.
\newblock {\em Pattern Formation and Dynamics in Nonequilibrium Systems}.
\newblock Cambridge University Press, 2009.

\bibitem[CH93]{CrossHohenberg1993}
M.~C. Cross and P.~C. Hohenberg.
\newblock Pattern formation outside of equilibrium.
\newblock {\em Reviews of Modern Physics; (United States)}, 65:3, 07 1993.

\bibitem[dBBCK08]{diBernardo2007}
M.~di~Bernardo, C.~J. Budd, A.~R. Champneys, and P.~Kowalczyk.
\newblock {\em Piecewise-Smooth Dynamical Systems: Theory and Applications},
  volume 163 of {\em Applied Mathematical Sciences}.
\newblock Springer, London, 2008.

\bibitem[Dev76]{MR402815}
R.~L. Devaney.
\newblock Reversible diffeomorphisms and flows.
\newblock {\em Trans. Amer. Math. Soc.}, 218:89--113, 1976.

\bibitem[DO21]{DoedelOldeman2021}
E.~J. Doedel and B.~E. Oldeman.
\newblock {\em AUTO-07P: Continuation and Bifurcation Software for Ordinary
  Differential Equations}.
\newblock Concordia University and McGill HPC Centre, Montreal, Canada, August
  2021.

\bibitem[HI11]{HaragusIooss}
M.~Haragus and G.~Iooss.
\newblock {\em Local Bifurcations, Center Manifolds, and Normal Forms in
  Infinite-Dimensional Dynamical Systems}.
\newblock Universitext. Springer, London, 2011.

\bibitem[Hos18]{Hosham2018}
H.~A. Hosham.
\newblock Bifurcations in four-dimensional switched systems.
\newblock {\em Advances in Difference Equations}, 2018(1):388, 2018.

\bibitem[KH11]{KuepperHosham2011}
T.~K{\"u}pper and H.~A. Hosham.
\newblock Reduction to invariant cones for non-smooth systems.
\newblock {\em Mathematics and Computers in Simulation}, 81(5):980--995, 2011.

\bibitem[KHW13]{KuepperHoshamWeiss2013}
T.~K{\"u}pper, H.~A. Hosham, and D.~Weiss.
\newblock Bifurcation for non-smooth dynamical systems via reduction methods.
\newblock In Andreas Johann, Hans-Peter Kruse, Florian Rupp, and Stephan
  Schmitz, editors, {\em Recent Trends in Dynamical Systems}, pages 79--105,
  Basel, 2013. Springer Basel.

\bibitem[KMM{\etalchar{+}}23]{KatsimigaMistakidis2023}
G.~C. Katsimiga, S.~I. Mistakidis, B.~A. Malomed, D.~J. Frantzeskakis,
  R.~Carretero-Gonzalez, and P.~G. Kevrekidis.
\newblock Interactions and dynamics of one-dimensional droplets, bubbles and
  kinks.
\newblock {\em Condensed Matter}, 8(3), 2023.

\bibitem[LBH24]{BettiniCenedeseHaller2024}
M.~Cenedese L.~Bettini and G.~Haller.
\newblock Model reduction to spectral submanifolds in piecewise smooth
  dynamical systems.
\newblock {\em International Journal of Non-Linear Mechanics}, 163:104753,
  2024.

\bibitem[PRV22]{PonceRosVela2022}
E.~Ponce, J.~Ros, and E.~Vela.
\newblock {\em Bifurcations in continuous piecewise linear differential
  systems---applications to low-dimensional electronic oscillators}, volume~7
  of {\em RSME Springer Series}.
\newblock Springer, Cham, 2022.

\bibitem[PRY23]{PruggerRademacherYang2023}
A.~Prugger, J.~D.~M. Rademacher, and J.~Yang.
\newblock Rotating shallow water equations with bottom drag: bifurcations and
  growth due to kinetic energy backscatter.
\newblock {\em SIAM J. Appl. Dyn. Syst.}, 22(3):2490--2526, 2023.

\bibitem[SH77]{SwiftHohenberg1977}
J.~Swift and P.~C. Hohenberg.
\newblock Hydrodynamic fluctuations at the convective instability.
\newblock {\em Physical Review A}, 15(1):319--328, 1977.

\bibitem[Sim22]{Simpson2022}
D.~J.~W. Simpson.
\newblock Twenty {H}opf-like bifurcations in piecewise-smooth dynamical
  systems.
\newblock {\em Phys. Rep.}, 970:1--80, 2022.

\bibitem[SMCGK23]{SaqlainMithun2023}
S.~Saqlain, T.~Mithun, R.~Carretero-Gonz\'alez, and P.~G. Kevrekidis.
\newblock Dragging a defect in a droplet bose-einstein condensate.
\newblock {\em Phys. Rev. A}, 107:033310, Mar 2023.

\bibitem[SZR20]{ZazoRdemacher2020}
M.~Steinherr~Zazo and J.~D.~M. Rademacher.
\newblock Lyapunov coefficients for {H}opf bifurcations in systems with
  piecewise smooth nonlinearity.
\newblock {\em SIAM J. Appl. Dyn. Syst.}, 19(4):2847--2886, 2020.

\bibitem[SZR23]{ZazoRademacher2023}
M.~Steinherr~Zazo and J.~D.~M. Rademacher.
\newblock Bifurcation control for a ship maneuvering model with nonsmooth
  nonlinearities.
\newblock {\em SIAM J. Control Optim.}, 61(1):225--251, 2023.

\bibitem[WC99]{WoodsChampneys1999}
P.~D. Woods and A.~R. Champneys.
\newblock Heteroclinic tangles and homoclinic snaking in the unfolding of a
  degenerate reversible {H}amiltonian-{H}opf bifurcation.
\newblock {\em Phys. D}, 129(3-4):147--170, 1999.

\bibitem[ZTH19]{ZemskovTsyganov2019}
E.~P. Zemskov, M.~A. Tsyganov, and W.~Horsthemke.
\newblock Multifront regime of a piecewise-linear fitzhugh-nagumo model with
  cross diffusion.
\newblock {\em Phys. Rev. E}, 99:062214, Jun 2019.

\end{thebibliography}

\end{document}